\documentclass[a4paper,10pt,notitlepage]{article}
\usepackage{amsmath}
\usepackage{amssymb}
\usepackage{xstring}%For Automatic Detection of reference prefixes.
\usepackage[english]{babel}
\usepackage{hyperref}
\usepackage{amsthm}
\usepackage{caption}
\usepackage{subcaption}
\usepackage{makecell}
\usepackage{tablefootnote}
\usepackage[boxed]{algorithm2e}
\usepackage{xcolor} %For ToDo Box
\usepackage{tikz} %For tikz graphics
\usepackage{pgfplots} %For tikz graphics
\pgfplotsset{compat=1.14}
\usepackage{graphicx}
\usepackage[a4paper,left=20mm,right=20mm,top=20mm,bottom=20mm]{geometry}
\newcommand{\refP}[1]{%
	\def\InputString{#1}%
	\IfBeginWith{\InputString}{Equation}{%
		(\ref{#1})}{%
	\IfBeginWith{\InputString}{Section}{%
		Section \ref{#1}}{%
	\IfBeginWith{\InputString}{Subsection}{%
		Subsection \ref{#1}}{%
	\IfBeginWith{\InputString}{Chapter}{%
		Chapter \ref{#1}}{%
	\IfBeginWith{\InputString}{Subsubsection}{%
		Subsubsection \ref{#1}}{%
	\IfBeginWith{\InputString}{Problem}{%
		Problem (\ref{#1})}{%
	\IfBeginWith{\InputString}{Property}{%
		Property (\ref{#1})}{%
	\IfBeginWith{\InputString}{Condition}{%
		Condition \ref{#1}}{%
	\IfBeginWith{\InputString}{Statement}{%
		Statement \ref{#1}}{%
	\IfBeginWith{\InputString}{Algorithm}{%
		Algorithm \ref{#1}}{%
	\IfBeginWith{\InputString}{Figure}{%
		Figure \ref{#1}}{%
	\IfBeginWith{\InputString}{Table}{%
		Table \ref{#1}}{%
	\IfBeginWith{\InputString}{Question}{%
		Question (\ref{#1})}{%
	\IfBeginWith{\InputString}{Footnote}{%
		Footnote \ref{#1}}{%
		\ref{#1}}}}}}}}}}}}}}}%
}
\newcommand{\thref}[1]{%
	\def\InputString{#1}%
	\IfBeginWith{\InputString}{Definition}{%
		Definition \ref{#1}}{%
	\IfBeginWith{\InputString}{Theorem}{%
		Theorem \ref{#1}}{%
	\IfBeginWith{\InputString}{Proposition}{%
		Proposition \ref{#1}}{%
	\IfBeginWith{\InputString}{Lemma}{%
		Lemma \ref{#1}}{%
	\IfBeginWith{\InputString}{Corollary}{%
		Corollary \ref{#1}}{%
	\IfBeginWith{\InputString}{Remark}{%
		Remark \ref{#1}}{%
	\IfBeginWith{\InputString}{Citedresult}{%
		Cited Result \ref{#1}}{%
	\IfBeginWith{\InputString}{Example}{%
		Example \ref{#1}}{%
	\IfBeginWith{\InputString}{Conjecture}{%
		Conjecture \ref{#1}}{%
	\IfBeginWith{\InputString}{Question}{%
		Question \ref{#1}}{%
	\IfBeginWith{\InputString}{Experiment}{%
		Experiment \ref{#1}}{%
	\IfBeginWith{\InputString}{Goal}{%
		Goal \ref{#1}}{%
		\ref{#1}}}}}}}}}}}}}%
}

\definecolor{TodoRed}{RGB}{150,50,0}

\newtheorem{CounterTheorem}{}[section]

\newtheorem{Definition}[CounterTheorem]{Definition}
\newtheorem{Theorem}[CounterTheorem]{Theorem}
\newtheorem{Proposition}[CounterTheorem]{Proposition}
\newtheorem{Lemma}[CounterTheorem]{Lemma}
\newtheorem{Corollary}[CounterTheorem]{Corollary}
\newtheorem{Remark}[CounterTheorem]{Remark}

\newtheorem{Example}[CounterTheorem]{Example}

\newtheorem{AlgorithmEnvirorementForNTheorem}[CounterTheorem]{Algorithm}

\renewcommand{\Vec}[1]{\mathbf{#1}}

\newcommand{\aVec}{\Vec{a}}
\newcommand{\bVec}{\Vec{b}}

\newcommand{\eVec}{\Vec{e}}

\newcommand{\pVec}{\Vec{p}}

\newcommand{\tVec}{\Vec{t}}

\newcommand{\vVec}{\Vec{v}}
\newcommand{\wVec}{\Vec{w}}
\newcommand{\xVec}{\Vec{x}}
\newcommand{\yVec}{\Vec{y}}
\newcommand{\zVec}{\Vec{z}}

\newcommand{\Mat}[1]{\mathbf{#1}}

\newcommand{\AMat}{\Mat{A}}
\newcommand{\BMat}{\Mat{B}}

\newcommand{\TextForAll}{\hspace{2pt} \text{ for all } \hspace{2pt}}

\newcommand{\TextIf}{\hspace{2pt}\text{ if }\hspace{2pt}}
\newcommand{\TextWith}{\hspace{2pt}\text{ with }\hspace{2pt}}
\newcommand{\TextAnd}{\hspace{2pt}\text{ and }\hspace{2pt}}

\newcommand{\ZeroSet}{\left\{0\right\}}

\newcommand{\abs}[1]{\left|#1\right|}

\newcommand{\norm}[1]{\left\|#1\right\|}
\newcommand{\scprod}[2]{\langle#1,#2\rangle}

\newcommand{\ProjToIndex}[2]{\left.#2\right|_{#1}}

\newcommand{\SetOf}[1]{\left[#1\right]}

\newcommand{\argmin}[1]{\underset{#1}{\textnormal{argmin}}\hspace{1pt}}

\def\AddBasicFunction#1#2{
	\expandafter\def\csname #1\endcsname##1{
		\def\InputString{##1}
		\def\CheckString{}%I define the String \CheckString to be an empty String.
		\ifx\InputString\CheckString %Checks if the Input is an Empty String.
			#2
		\else
			#2 \left(##1\right)
		\fi
	}
}
\AddBasicFunction{Exp}{\exp}
\AddBasicFunction{Cos}{\cos}
\AddBasicFunction{Sin}{\sin}
\AddBasicFunction{Tan}{\tan}
\AddBasicFunction{Ln}{\ln}
\AddBasicFunction{Log}{\log}
\AddBasicFunction{ArcCos}{\arccos}
\AddBasicFunction{ArcSin}{\arcsin}
\AddBasicFunction{ArcTan}{\arctan}
\AddBasicFunction{Cosh}{\cosh}
\AddBasicFunction{Sinh}{\sinh}
\AddBasicFunction{Tanh}{\tanh}
\AddBasicFunction{ArcCosh}{\textnormal{arccosh}}
\AddBasicFunction{ArcSinh}{\textnormal{arcsinh}}
\AddBasicFunction{ArcTanh}{\textnormal{arctanh}}

\AddBasicFunction{Kernel}{\ker}
\AddBasicFunction{Rank}{\textnormal{rank}}
\AddBasicFunction{Range}{\textnormal{Ran}}

\AddBasicFunction{supp}{\textnormal{supp}}
\newcommand{\conv}[1]{\textnormal{conv}\left(#1\right)}

\AddBasicFunction{sgn}{\textnormal{sgn}}

\newcommand{\ExpE}{\mathrm{e}}

\newcommand{\SetSize}[1]{\#\left(#1\right)}

\title{
	Non-negative Sparse Recovery at Minimal Sampling Rate
}
\date{\today}
\author{
	Hendrik Bernd Zarucha
	\thanks{
		Communications and Information Theory Group,
		Technische Universtität Berlin, Berlin,
		petersen@tu-berlin.de,
	}
	and
	Peter Jung
	\thanks{
		Communications and Information Theory Group,
		Technische Universtität Berlin, Berlin,
		peter.jung@tu-berlin.de,
		and German Aerospace Center (DLR)
	}
}
\newcommand{\OrderOf}[1]{\mathcal{O}\left(#1\right)}
\newcommand{\normLHSsign}{L}
\newcommand{\normRHSsign}{R}
\newcommand{\normRHS}[1]{\norm{#1}_\normRHSsign}
\newcommand{\normLHS}[1]{\norm{#1}_\normLHSsign}
\AddBasicFunction{p}{p}
\newcommand{\Cone}[1]{\textnormal{cone}\left(#1\right)}

\begin{document}
	\maketitle
	\begin{abstract}
	It is known that sparse recovery is possible
	if the number of measurements is in the order of the sparsity,
	but the corresponding decoders either lack polynomial decoding time or
	robustness to noise.	
	Commonly, decoders that rely on a null space property are being used.
	These achieve polynomial time decoding
	and are robust to additive noise but pay the price
	by requiring more measurements. The non-negative least residual
	has been established as such a decoder for non-negative recovery.
	A new equivalent condition for uniform, robust recovery of non-negative sparse
	vectors with the non-negative least residual that is not based
	on null space properties is introduced. It is shown that the number of measurements
	for this equivalent condition only needs to be in the order of the sparsity.
	Further, it is explained why the robustness to additive noise is
	similar, but not equal, to the robustness of
	decoders based on null space properties.
\end{abstract}
	\section{Introduction}
The goal of compressed sensing is to design a measurement matrix
$\AMat\in\mathbb{R}^{M\times N}$ and a decoder $Q:\mathbb{R}^M\rightarrow\mathbb{R}^N$ such that $Q\left(\AMat\xVec\right)=\xVec$ for all $S$-sparse $\xVec$ with $M$ being as small as possible.
By choosing $\ell_0$-minimization as decoder and a suitable matrix
one can show that the number of measurements can be as small as $M=2S$
\cite[Theorem~2.14]{IntroductionCS}.
However, $\ell_0$-minimization is NP-hard 
and can not be performed in polynomial time in general \cite[Theorem~2.17]{IntroductionCS}.
Only if $\AMat$ is a Vandermonde matrix,
it is known that $\ell_0$-minimization can be performed in
a polynomial time using Pronys method \cite[Section~5]{Prony_method}\cite[Chapter~4.4]{Prony_method_diss}\cite[Theorem~2.15]{IntroductionCS}.
Numerical evidence suggests that Pronys method is infeasibly sensitive
to noise \cite[Exercise~2.8]{IntroductionCS}. To the best of the authors knowledge
no recovery guarantees are known that suggest robustness to noise
for $\ell_0$-minimization.
Since in most cases noise corrupts the measurements
$\yVec=\AMat\xVec+\eVec$, there is no feasible method for sparse recovery
whenever $M\in\OrderOf{S}$.\par
To overcome this problem one can consider matrices with a null space property.
Such matrices exist if $M\in\OrderOf{S\Ln{\ExpE\frac{N}{S}}}$ \cite[Theorem~9.27~+~Theorem~6.13]{IntroductionCS},
which scales slightly worse than $2S$.
In particular, if $\AMat$ has a null space property,
there exist several decoders (for instance the basis pursuit denoising \cite[Chapter~4]{IntroductionCS}
and the rLasso \cite{rLasso})
that not only obey $Q\left(\AMat\xVec\right)=\xVec$ for all $S$-sparse $\xVec$
but the estimation error is also bounded by a linear function
in the norm of the noise vector. This latter property is called robustness. If the constant determining the linear function is large, the robustness is called bad. Restricted to non-negative $\xVec$,
the non-negative least residual can achieve the same error bounds
if $\AMat$ has a null space property and a further minor property
\cite{NNLS,NNLAD}.
\par
In this work the signed kernel condition is introduced and its equivalence
to exact recovery of all non-negative $S$-sparse signals
with the non-negative least residual is proven.
Further, it is shown that under a signed kernel condition
the estimation error is also bounded by a linear function
in the norm of the noise.
Another result shows that Vandermonde matrices have the
signed kernel condition with the optimal number of measurements $M\geq 2S+1$
and that the underlying problem can be solved in a polynomial time because it
is a linear program.
In other words the open case in \refP{Table:solved_problems} is filled
at least in the case of non-negative recovery.
\begin{table}[ht]
	\centering
	\caption{Solved Uniform Recovery Problems in Compressed Sensing}
	\label{Table:solved_problems}
	\begin{tabular}{|c|c|c|c|}
		\hline
			\makecell{
				sparse
				\\ recovery
			}
			& robust
			& polynomial time
			& \makecell{
				robust and
				\\ polynomial time
			}
		\\\hline
			$M\in\mathcal{O}\left(S\right)$ 
			& \makecell{
				$\ell_0$-constrained
				\\ residual minimization
				\tablefootnote{
					To the best of the authors knowledge this statement is well known. If not, readers are free to apply this work with the Kruskal Rank Case to get the result. However, the authors do not claim originality over this result!
				}
			}
			& \makecell{
				Pronys method
				\\ \cite[Theorem~2.15]{IntroductionCS}
			}
			& ?
		\\\hline
			$M\in\mathcal{O}\left(S\Ln{\ExpE\frac{N}{S}}\right)$
			& \makecell{
				basis pursuit denoising
				\\ \cite[Chapter~4]{IntroductionCS}
			}
			& \makecell{
				basis pursuit denoising
				\\ \cite[Chapter~13]{IntroductionCS}
			}
			& \makecell{
				basis pursuit denoising
				\\ \cite[Chapter~13]{IntroductionCS}
			}
		\\\hline
	\end{tabular}
\end{table}
To the best of the authors knowledge this is the first construction
of a matrix with a polynomial time decoder that achieve uniform and
robust recovery guarantees with $M\in\OrderOf{S}$.
It is also shown that with complex measurement matrices
$M\geq S+1$ measurements are sufficient for a signed kernel condition.
Further, a new explanation for the difference in stable recovery between
the $\mathcal{O}\left(S\right)$ and
$\mathcal{O}\left(S\Ln{\ExpE\frac{N}{S}}\right)$ case is given.
Lastly, a method to calculate the robustness constant (not in polynomial time)
is given.
	\section*{Notation}
Given $N\in\mathbb{N}$ denote the indices up to $N$ as
$\SetOf{N}:=\left\{1,\dots,N\right\}$. Given a set $T\subset\SetOf{N}$
let $\SetSize{T}$ be the number of elements in $T$.
Over the course of this work the body $\mathbb{K}$ denotes either the real numbers $\mathbb{R}$
or the complex numbers $\mathbb{C}$. $\mathbb{K}$ will be specified whenever it is required.
A vector $\xVec\in\mathbb{K}^N$ is called $S$-sparse if
$\norm{\xVec}_0:=\SetSize{\supp{\xVec}}\leq S$ where $\supp{\xVec}:=\left\{n\in\SetOf{N}:\xVec_n\neq 0\right\}$.
The set of sparse vectors is denoted by
$\Sigma_S^N:=\left\{\zVec\in\mathbb{K}^N:\norm{\zVec}_0\leq S\right\}$
and set of the non-negative vectors by
$\mathbb{R}^N_+
	:=\left\{\zVec\in\mathbb{R}^N:\zVec_n\geq 0 \TextForAll n\in\SetOf{N}\right\}$.
Within this work two norms are required and denoted by $\normLHS{\cdot}$ and $\normRHS{\cdot}$.
These norms do not need to be $\ell_p$-norms and $L,R$ are to be understand as a part of the defined
symbol and not a variable.
$\ell_p$-norms will always be denoted by $\norm{\cdot}_p$ with a variable denoted by
$p$, $q$ or an element of $p\in\left[1,\infty\right]$ to distinguish them clearly from
$\normLHS{\cdot}$ and $\normRHS{\cdot}$.
Given a norm $\norm{\cdot}$ the unit sphere and ball are denoted by
$\mathbb{S}^{N-1}_{\norm{\cdot}}
	:=\left\{\zVec\in\mathbb{K}^N:\norm{\zVec}=1\right\}$ and
$\mathbb{B}^N_{\norm{\cdot}}
	:=\left\{\zVec\in\mathbb{K}^N:\norm{\zVec}\leq 1\right\}$ respectively.
Given $\xVec\in\mathbb{K}^N$ and $T\subset\SetOf{N}$
the vector which replaces the entries of $\xVec$ with coordinates not in $T$
by 0 is denoted by $\ProjToIndex{T}{\xVec}$.
The kernel or null space of a matrix $\AMat\in\mathbb{K}^{M\times N}$
is denoted by $\Kernel{\AMat}:=\left\{\zVec\in\mathbb{K}^N:\AMat\xVec=0\right\}$
which is depending on the choice of $\mathbb{K}$!
Given a set $F$ the cone containing $F$ is denoted by
$\Cone{F}:=\left\{\alpha\xVec:\xVec\in F,\alpha\geq 0\right\}$
and the convex hull is denoted by $\conv{F}$.
Given a function $f$ and a set of feasible elements $F$
the set of minimizers of the corresponding optimization problem
is denoted by
$\argmin{z\in F}f(z):=\left\{x^\#\in F: f(x^\#)=\inf_{z\in C}f(z)\right\}$.
	\section{Robustness and Recovery Conditions}\label{Section:robustness_recovery}
In this work $C\subset\mathbb{K}^N$ is a set of vectors
that are supposed to be recoverable from as few as possible noisy, linear measurements
$\yVec=\AMat\xVec+\eVec$. Given some norm $\normRHS{\cdot}$ a simple decoding
approach would try to solve $\argmin{\zVec\in C}\normRHS{\AMat\zVec-\yVec}$.
For a possibly larger, closed set of candidates $F$ with $C\subset F\subset \mathbb{K}^N$
one can try to solve the relaxed problem
$\xVec^\#\in\argmin{\zVec\in F}\normRHS{\AMat\zVec-\yVec}$ instead.
This might be easier if for instance
$F$ is convex but $C$ is of a combinatorial nature.
This problem formulation can potentially cover uniform recovery guarantees
from the following cases.
\begin{itemize}
	\item
		The Kruskal Rank case:
		$C=\Sigma_S^N$ and $F=\Sigma_S^N$.
	\item
		The Signed Kernel Condition case:
		$C=\Sigma_S^N\cap\mathbb{R}^N_+$ and $F=\mathbb{R}^N_+$.
		This will be introduced later.
	\item
		The Null Space Property case:
		$C=\eta\left(\Sigma_S^N\cap\mathbb{S}^{N-1}_{\norm{\cdot}_1}\right)$
		and $F=\eta\mathbb{B}^{N}_{\norm{\cdot}_1}$,
		where $\eta>0$.
	\item
		The Outwardly Neighborly Polytope case:
		$C=\eta\left(\Sigma_S^N\cap\mathbb{S}^{N-1}_{\norm{\cdot}_1}\cap\mathbb{R}_+^N\right)$
		and $F=\eta\left(\mathbb{B}^{N}_{\norm{\cdot}_1}\cap\mathbb{R}_+^N\right)$,
		where $\eta>0$.
	\item
		The Finite Code Book case:
		$\SetSize{C}<\infty$ and $F=\conv{C}$.
\end{itemize}
The following constants are central to study recovery and robustness. 
\begin{Definition}\label{Definition:tau_kappa}
	Let $C\subset F\subset\mathbb{K}^N$,
	$\normLHS{\cdot}$ a norm on $\mathbb{K}^N$,
	$\normRHS{\cdot}$ a norm on $\mathbb{K}^M$
	and $\AMat\in\mathbb{K}^{M\times N}$.
	Then,
	\begin{align}
		\nonumber
			\tau\left(\AMat\right)
		:=
			\inf_{\zVec\in F,\xVec\in C:\zVec\neq\xVec}
			\frac{
				\normRHS{\AMat\left(\zVec-\xVec\right)}
			}{
				\normLHS{\zVec-\xVec}
			}
		=
			\inf_{\vVec\in F-C:\vVec\neq 0}
			\frac{
				\normRHS{\AMat\vVec}
			}{
				\normLHS{\vVec}
			}
		=
			\inf_{\vVec\in \Cone{F-C}:\vVec\neq 0}
			\frac{
				\normRHS{\AMat\vVec}
			}{
				\normLHS{\vVec}
			}
	\end{align}
	is called robustness constant and
	\begin{align}
		\nonumber
			\kappa\left(\AMat\right)
		:=
			\sup_{\zVec\in F,\xVec\in C:\zVec\neq\xVec}
			\frac{
				\normRHS{\AMat\left(\zVec-\xVec\right)}
			}{
				\normLHS{\zVec-\xVec}
			}
		=
			\sup_{\vVec\in F-C:\vVec\neq 0}
			\frac{
				\normRHS{\AMat\vVec}
			}{
				\normLHS{\vVec}
			}
		=
			\sup_{\vVec\in \Cone{F-C}:\vVec\neq 0}
			\frac{
				\normRHS{\AMat\vVec}
			}{
				\normLHS{\vVec}
			}
	\end{align}
	is called normalization constant.
\end{Definition}
The robustness constant is closely related to recovery and robustness.
\begin{Theorem}[Robustness and Recovery]
	\label{Theorem:robustness_and_recovery}
	Let $C\subset F\subset\mathbb{K}^N$,
	$\normLHS{\cdot}$ a norm on $\mathbb{K}^N$,
	$\normRHS{\cdot}$ a norm on $\mathbb{K}^M$
	and $\AMat\in\mathbb{K}^{M\times N}$.
	Consider the following conditions:
	\begin{enumerate}
		\item
		\label{Condition:Theorem:robustness_and_recovery:tau}
			$\tau\left(\AMat\right)>0$.
		\item
		\label{Condition:Theorem:robustness_and_recovery:robust}
			There exists a constant $D>0$ such that one has
			\begin{align}
				\nonumber
					\normLHS{\xVec^\#-\xVec}
				\leq
					D\normRHS{\yVec-\AMat\xVec}
				\TextForAll
					\xVec\in C,\yVec\in\mathbb{K}^M,
					\xVec^\#\in\argmin{\zVec\in F}\normRHS{\AMat\zVec-\yVec}.
			\end{align}
		\item
		\label{Condition:Theorem:robustness_and_recovery:kernel}
			$\left(F-C\right)\cap\Kernel{\AMat}=\left\{0\right\}$.
		\item
		\label{Condition:Theorem:robustness_and_recovery:recovery}
			For all $\xVec\in C$ one has
			\begin{align}
					\nonumber
					\left\{\xVec\right\}
				=
					\argmin{\zVec\in F}\normRHS{\AMat\zVec-\AMat\xVec}.
			\end{align}
	\end{enumerate}
	Then, the following statements are true:
	\begin{enumerate}
		\item
		\label{Statement:Theorem:robustness_and_recovery:tau_robust}
			\refP{Condition:Theorem:robustness_and_recovery:tau}
			and
			\refP{Condition:Theorem:robustness_and_recovery:robust}
			are equivalent.
			Moreover
			\refP{Condition:Theorem:robustness_and_recovery:tau}
			yields
			\refP{Condition:Theorem:robustness_and_recovery:robust}
			with $D=\frac{2}{\tau\left(\AMat\right)}$
			and
			\refP{Condition:Theorem:robustness_and_recovery:robust}
			yields
			\refP{Condition:Theorem:robustness_and_recovery:tau}
			with $D\geq\frac{1}{\tau\left(\AMat\right)}$.
		\item
		\label{Statement:Theorem:robustness_and_recovery:kernel_recovery}
			\refP{Condition:Theorem:robustness_and_recovery:kernel}
			and
			\refP{Condition:Theorem:robustness_and_recovery:recovery}
			are equivalent.
		\item
		\label{Statement:Theorem:robustness_and_recovery:robust=>recovery}
			\refP{Condition:Theorem:robustness_and_recovery:tau}
			and
			\refP{Condition:Theorem:robustness_and_recovery:robust}
			imply
			\refP{Condition:Theorem:robustness_and_recovery:kernel}
			and
			\refP{Condition:Theorem:robustness_and_recovery:recovery}.
		\item
		\label{Statement:Theorem:robustness_and_recovery:recovery=>robust}
			Let additionally $\Cone{F-C}$ be closed. Then,
			\refP{Condition:Theorem:robustness_and_recovery:kernel}
			and
			\refP{Condition:Theorem:robustness_and_recovery:recovery}
			imply
			\refP{Condition:Theorem:robustness_and_recovery:tau}
			and
			\refP{Condition:Theorem:robustness_and_recovery:robust}.
		\item
		\label{Statement:Theorem:robustness_and_recovery:equivalent}
			Let additionally $\Cone{F-C}$ be closed. Then,
			\refP{Condition:Theorem:robustness_and_recovery:tau},
			\refP{Condition:Theorem:robustness_and_recovery:robust},
			\refP{Condition:Theorem:robustness_and_recovery:kernel}
			and
			\refP{Condition:Theorem:robustness_and_recovery:recovery}
			are all equivalent.
	\end{enumerate}
\end{Theorem}
\begin{proof}
	\refP{Statement:Theorem:robustness_and_recovery:tau_robust}:
	\refP{Condition:Theorem:robustness_and_recovery:tau}
	$\Rightarrow$
	\refP{Condition:Theorem:robustness_and_recovery:robust}:
	Let $\xVec\in C$, $\yVec\in \mathbb{R}^M$ and
	$\xVec^\#\in\argmin{\zVec\in F}\normRHS{\AMat\zVec-\yVec}$.
	By the definition of $\tau\left(\AMat\right)$ and $\xVec^\#$
	being a minimizer and $\xVec\in C\subset F$ being feasible one gets
	\begin{align}
		\nonumber
			\normLHS{\xVec^\#-\xVec}
		\leq
			\frac{1}{\tau\left(\AMat\right)}\normRHS{\AMat\xVec^\#-\AMat\xVec}
		\leq
			\frac{1}{\tau\left(\AMat\right)}
			\left(
				\normRHS{\AMat\xVec^\#-\yVec}
				+\normRHS{\AMat\xVec-\yVec}
			\right)
		\leq
			\frac{2}{\tau\left(\AMat\right)}
			\normRHS{\AMat\xVec-\yVec}.
	\end{align}
	This is \refP{Condition:Theorem:robustness_and_recovery:robust}
	with $D=\frac{2}{\tau\left(\AMat\right)}$.
	\par
	\refP{Condition:Theorem:robustness_and_recovery:robust}
	$\Rightarrow$
	\refP{Condition:Theorem:robustness_and_recovery:tau}:
	There exists $\left(\zVec^k\right)_{k\in\mathbb{N}}\subset F,\left(\xVec^k\right)_{k\in\mathbb{N}}\subset C$ such that
	$\zVec^k\neq\xVec^k$ for all $k\in\mathbb{N}$ and
	\begin{align}
		\label{Equation:Theorem:robustness_and_recovery:eq1}
			\tau\left(\AMat\right)
		=
			\lim_{k\rightarrow \infty}
			\frac{
				\normRHS{\AMat\left(\xVec^k-\zVec^k\right)}
			}{
				\normLHS{\xVec^k-\zVec^k}
			}.
	\end{align}
	By setting $\yVec^k:=\AMat\zVec^k$ it follows that
	$\zVec^k$ is a minimizer of
	$\min_{\zVec\in F}\normRHS{\AMat\zVec-\yVec^k}$.
	Applying \refP{Condition:Theorem:robustness_and_recovery:robust}
	to \eqref{Equation:Theorem:robustness_and_recovery:eq1} yields
	\begin{align}
		\nonumber
			\tau\left(\AMat\right)
		=
			\lim_{k\rightarrow \infty}
			\frac{
				\normRHS{\AMat\xVec^k-\yVec^k}
			}{
				\normLHS{\xVec^k-\zVec^k}
			}
		\geq
			\frac{1}{D}
		>
			0.
	\end{align}
	\par
	\refP{Statement:Theorem:robustness_and_recovery:kernel_recovery}:
	\refP{Condition:Theorem:robustness_and_recovery:recovery}
	$\Rightarrow$
	\refP{Condition:Theorem:robustness_and_recovery:kernel}:
	Let
	$\vVec\in \left(F-C\right)\cap\Kernel{\AMat}$.
	Then, there exists $\zVec\in F,\xVec\in C\subset F$ such that
	$\vVec=\zVec-\xVec$ and it follows that
	$\AMat\zVec=\AMat\xVec+\AMat\vVec=\AMat\xVec$.
	As a consequence $\zVec$ is a minimizer of
	$\min_{\zVec'\in F}\normRHS{\AMat\zVec'-\AMat\xVec}$.
	By \refP{Condition:Theorem:robustness_and_recovery:recovery}
	one gets $\zVec=\xVec$
	and thus $\left(F-C\right)\cap\Kernel{\AMat}=\left\{0\right\}$.
	\par
	\refP{Condition:Theorem:robustness_and_recovery:kernel}
	$\Rightarrow$
	\refP{Condition:Theorem:robustness_and_recovery:recovery}:
	Let $\xVec\in C$ and
	$\zVec\in\argmin{\zVec'\in F}\normRHS{\AMat\zVec'-\AMat\xVec}$.
	Since $\xVec\in C\subset F$ has objective value of $0$,
	also $\zVec$ needs to have objective value of $0$.
	Thus $\AMat\zVec=\AMat\xVec$ and $\zVec-\xVec\in \left(F-C\right)\cap\Kernel{\AMat}$.
	By \refP{Condition:Theorem:robustness_and_recovery:kernel}
	one gets $\zVec=\xVec$ is the unique minimizer.
	\par
	\refP{Statement:Theorem:robustness_and_recovery:robust=>recovery}:
	\refP{Condition:Theorem:robustness_and_recovery:robust}
	$\Rightarrow$
	\refP{Condition:Theorem:robustness_and_recovery:recovery}:
	Let $\xVec\in C$. Since, $C\subset F$, $\xVec$ is indeed
	a minimizer of $\min_{\zVec\in F}\normRHS{\AMat\zVec-\AMat\xVec}$.
	Now let $\xVec^\#$ be any minimizer of
	$\min_{\zVec\in F}\normRHS{\AMat\zVec-\AMat\xVec}$.
	Applying \refP{Condition:Theorem:robustness_and_recovery:robust}
	with $\yVec:=\AMat\xVec$ yields
	$\normLHS{\xVec^\#-\xVec}\leq D\normRHS{0}$.
	Thus, $\xVec$ is the unique minimizer.
	\par
	\refP{Statement:Theorem:robustness_and_recovery:recovery=>robust}:
	Now let additionally $\Cone{F-C}$ be closed. In order to show
	\refP{Condition:Theorem:robustness_and_recovery:kernel}
	$\Rightarrow$
	\refP{Condition:Theorem:robustness_and_recovery:tau}
	assume the logical inverse and let $\tau\left(\AMat\right)=0$.
	Hence, there exist sequences
	$\left({\zVec}^k\right)_{k\in\mathbb{N}}\subset F$
	and $\left(\xVec^k\right)_{k\in\mathbb{N}}\subset C$
	such that
	${\zVec}^k\neq\xVec^k$ for all $k\in\mathbb{N}$ and
	\begin{align}\label{Equation:Theorem:robustness_and_recovery:limit}
			\lim_{k\rightarrow\infty}
			\normRHS{
				\AMat\frac{{\zVec}^k-\xVec^k}{\normLHS{{\zVec}^k-\xVec^k}}
			}
		=
			0.
	\end{align}
	Then, the sequence $\left(\vVec^k\right)_{k\in\mathbb{N}}$ defined by
	$\vVec^k:=\frac{\zVec^k-\xVec^k}{\normLHS{\zVec^k-\xVec^k}}\in 
		\Cone{F-C}\cap \mathbb{S}^{N-1}_{\normLHS{\cdot}}$ for all $k\in\mathbb{N}$
	has a convergent subsequence converging
	to some $\vVec$. Since $\Cone{F-C}$ is closed,
	$\vVec\in\Cone{F-C}\cap\mathbb{S}^{N-1}_{\normLHS{\cdot}}$.
	By \eqref{Equation:Theorem:robustness_and_recovery:limit} one gets
	$\normRHS{\AMat\vVec}=0$ and thus, $\vVec\in\Kernel{\AMat}$.
	Due to the normalization
	$\vVec\in\Cone{F-C}\cap\Kernel{\AMat}\setminus\ZeroSet$
	and hence
	$\left(F-C\right)\cap\Kernel{\AMat}\neq\ZeroSet$.
	\par
	\refP{Statement:Theorem:robustness_and_recovery:equivalent}
	follows from
	\refP{Statement:Theorem:robustness_and_recovery:tau_robust},
	\refP{Statement:Theorem:robustness_and_recovery:kernel_recovery},
	\refP{Statement:Theorem:robustness_and_recovery:robust=>recovery}
	and
	\refP{Statement:Theorem:robustness_and_recovery:recovery=>robust}.
\end{proof}
\refP{Condition:Theorem:robustness_and_recovery:kernel} determines whether uniform
recovery is possible.
The robustness constant determines whether or not uniform and robust
recovery is possible and gives insight how robust the recovery is.
Due to \refP{Statement:Theorem:robustness_and_recovery:tau_robust}
it is guaranteed that $\frac{2}{\tau\left(\AMat\right)}$
is a constant for a robust recovery guarantee and that this
constant is optimal up to a factor of $2$. It is thus, sufficient
to study the robustness constant and not required
to study general constants for robust recovery guarantees.
If $\Cone{F-C}$ is closed, recovery is equivalent to robust recovery.
The robustness comes for free in this case.
The following theorem is concerned with small perturbations
of measurement matrices with recovery guarantees.
\begin{Theorem}
	\label{Theorem:openness}
	Let $C\subset F\subset\mathbb{K}^N$,
	$\normLHS{\cdot}$ a norm on $\mathbb{K}^N$,
	$\normRHS{\cdot}$ a norm on $\mathbb{K}^M$.
	Then, the function $\AMat\mapsto\tau\left(\AMat\right)$ is continuous
	and the set of matrices with a robust recovery guarantee
	$\left\{\AMat\in\mathbb{K}^{M\times N}:\tau\left(\AMat\right)>0\right\}$
	is open.
	If additionally $\Cone{F-C}$ is closed, the set of matrices with a recovery guarantee
	$\big\{\AMat\in\mathbb{K}^{M\times N}:\left(F-C\right)\cap\Kernel{\AMat}=\ZeroSet\big\}$
	is open.
\end{Theorem}
\begin{proof}
	For this proof let
	$\norm{\AMat}_{L\mapsto R}:=\sup_{\vVec\in\mathbb{K}^N:\normLHS{\vVec}\leq 1}\normRHS{\AMat\vVec}$
	denote the operator norm of $\AMat$.
	Now let $\left(\AMat^k\right)_{k\in\mathbb{N}}\subset\mathbb{K}^{M\times N}$
	be a sequence such that $\lim_{k\rightarrow\infty}\AMat^k=\AMat$
	so that $\lim_{k\rightarrow\infty}\norm{\AMat^k-\AMat}_{L\mapsto R}=0$.
	For every $k$ there exists a sequence $\left(\vVec^{k,l}\right)_{l\in\mathbb{N}}\subset F-C$ with
	$\vVec^{k,l}\neq 0$ for all $l\in\mathbb{N}$ and
	\begin{align}
		\nonumber
			\tau\left(\AMat^k\right)
		=
			\lim_{l\rightarrow\infty}
			\frac{
				\normRHS{\AMat^k\vVec^{k,l}}
			}{
				\normLHS{\vVec^{k,l}}
			}.
	\end{align}
	Using this yields
	\begin{align}
		\nonumber
			\tau\left(\AMat\right)
		=&
			\inf_{\vVec\in F-C:\vVec\neq 0}
			\frac{\normRHS{\AMat\vVec}}{\normLHS{\vVec}}
		\leq
			\inf_{\vVec\in F-C:\vVec\neq 0}
			\frac{\normRHS{\left(\AMat-\AMat^k\right)\vVec}}{\normLHS{\vVec}}
			+\frac{\normRHS{\AMat^k\vVec}}{\normLHS{\vVec}}
		\\\leq&\nonumber
			\frac{\normRHS{\left(\AMat-\AMat^k\right)\vVec^{k,l}}}{\normLHS{\vVec^{k,l}}}
			+\frac{\normRHS{\AMat^k\vVec^{k,l}}}{\normLHS{\vVec^{k,l}}}
		\leq
			\sup_{\vVec\in\mathbb{K}^N:\normLHS{\vVec}\leq 1}
			\normRHS{\left(\AMat-\AMat^k\right)\vVec}
			+\frac{\normRHS{\AMat^k\vVec^{k,l}}}{\normLHS{\vVec^{k,l}}}
		\\=&\nonumber
			\norm{\AMat-\AMat^k}_{L\mapsto R}
			+\frac{\normRHS{\AMat^k\vVec^{k,l}}}{\normLHS{\vVec^{k,l}}}.
	\end{align}
	Applying two limits yields
	\begin{align}
		\label{Equation:Theorem:openness:ineq1}
			\tau\left(\AMat\right)
		\leq
			\lim_{k\rightarrow\infty}\lim_{l\rightarrow\infty}
			\norm{\AMat-\AMat^k}_{L\mapsto R}
			+\frac{\normRHS{\AMat^k\vVec^{k,l}}}{\normLHS{\vVec^{k,l}}}
		=		
			\lim_{k\rightarrow\infty}\norm{\AMat-\AMat^k}_{L\mapsto R}+\tau\left(\AMat^k\right)
		=		
			\lim_{k\rightarrow\infty}\tau\left(\AMat^k\right).
	\end{align}
	On the other hand there exists a sequence $\left(\vVec^{k}\right)_{k\in\mathbb{N}}\subset F-C$ with
	$\vVec^{k}\neq 0$ for all $k\in\mathbb{N}$ and
	\begin{align}
		\nonumber
			\tau\left(\AMat\right)
		=
			\lim_{k\rightarrow\infty}
			\frac{
				\normRHS{\AMat\vVec^{k}}
			}{
				\normLHS{\vVec^{k}}
			}.
	\end{align}
	Using this yields
	\begin{align}
		\nonumber
			\tau\left(\AMat^k\right)
		=&
			\inf_{\vVec\in F-C:\vVec\neq 0}
			\frac{\normRHS{\AMat^k\vVec}}{\normLHS{\vVec}}
		\leq
			\inf_{\vVec\in F-C:\vVec\neq 0}
			\frac{\normRHS{\left(\AMat^k-\AMat\right)\vVec}}{\normLHS{\vVec}}
			+\frac{\normRHS{\AMat\vVec}}{\normLHS{\vVec}}
		\\\leq&\nonumber
			\frac{\normRHS{\left(\AMat^k-\AMat\right)\vVec^{k}}}{\normLHS{\vVec^{k}}}
			+\frac{\normRHS{\AMat\vVec^{k}}}{\normLHS{\vVec^{k}}}
		\leq
			\sup_{\vVec\in\mathbb{K}^N:\normLHS{\vVec}\leq 1}
			\normRHS{\left(\AMat^k-\AMat\right)\vVec}
			+\frac{\normRHS{\AMat\vVec^{k}}}{\normLHS{\vVec^{k}}}
		\\=&\nonumber
			\norm{\AMat^k-\AMat}_{L\mapsto R}
			+\frac{\normRHS{\AMat\vVec^{k}}}{\normLHS{\vVec^{k}}}.
	\end{align}
	Applying a limit to this yields
	\begin{align}
		\nonumber
			\lim_{k\rightarrow\infty}\tau\left(\AMat^k\right)
		\leq
			\lim_{k\rightarrow\infty}
			\norm{\AMat^k-\AMat}_{L\mapsto R}
			+\frac{\normRHS{\AMat\vVec^{k}}}{\normLHS{\vVec^{k}}}
		=
			\tau\left(\AMat\right).
	\end{align}
	By this and \eqref{Equation:Theorem:openness:ineq1} 
	one gets $\lim_{k\rightarrow\infty}\tau\left(\AMat^k\right)=\tau\left(\AMat\right)$.
	Thus, the function $\AMat\mapsto\tau\left(\AMat\right)$ is continuous
	and the set
	$\left\{\AMat\in\mathbb{R}^{M\times N}:\tau\left(\AMat\right)>0\right\}$
	is the preimage of an open set of a continuous function and thus open.
	If additionally $\Cone{F-C}$ is closed, then
	\thref{Theorem:robustness_and_recovery} yields that
	$\left\{\AMat\in\mathbb{K}^{M\times N}:\left(F-C\right)\cap\Kernel{\AMat}=\ZeroSet\right\}
	=\left\{\AMat\in\mathbb{K}^{M\times N}:\tau\left(\AMat\right)>0\right\}$
	which in turn must be open as well.
\end{proof}
While the openness does not seem insightful, this result is incredibly important.
It guarantees that small perturbations of matrices with a robust recovery guarantee will also
have a similar good robust recovery guarantees. Since in general matrices can only be accurately
be represented up to machine precision in computer algebra systems,
such a property is of utmost importance.
Note that the condition that $\Cone{F-C}$ is closed is mandatory in
\refP{Statement:Theorem:robustness_and_recovery:recovery=>robust} of \thref{Theorem:robustness_and_recovery}
and in \thref{Theorem:openness}.
\begin{Example}
	Let $\mathbb{K}=\mathbb{R}$, $\AMat:=\begin{pmatrix}1&0\end{pmatrix}$, $C:=\ZeroSet$ and
	$F:=\left\{\xVec\in\mathbb{R}^2:\norm{\xVec-\begin{pmatrix}1\\0\end{pmatrix}}_2\leq 1\right\}$.
	Then,
	\begin{align}
		\nonumber
			\Cone{F-C}
		=
			\left\{
				\begin{pmatrix}\xVec_1\\\xVec_2\end{pmatrix}
				:\xVec_1>0
			\right\}
			\cup\ZeroSet
	\end{align}
	is not closed!
	\refP{Condition:Theorem:robustness_and_recovery:kernel}
	and
	\refP{Condition:Theorem:robustness_and_recovery:recovery}
	of \thref{Theorem:robustness_and_recovery}
	are fulfilled, but condition
	\refP{Condition:Theorem:robustness_and_recovery:tau}
	and
	\refP{Condition:Theorem:robustness_and_recovery:robust}
	of \thref{Theorem:robustness_and_recovery}
	are not fulfilled
	both of which can be seen by considering
	$\begin{pmatrix}1+\cos\left(\alpha\right)\\\sin\left(\alpha\right)\end{pmatrix}\in F-C$ for $\alpha\rightarrow\pi$.
	Further, the small perturbation of $\AMat$ given by
	$\AMat^\epsilon:=\begin{pmatrix}1&\epsilon\end{pmatrix}$ for any
	$\epsilon\in\mathbb{R}\setminus\ZeroSet$
	does not suffice
	\refP{Condition:Theorem:robustness_and_recovery:kernel}
	and
	\refP{Condition:Theorem:robustness_and_recovery:recovery}
	of \thref{Theorem:robustness_and_recovery}.
\end{Example}
Due to \thref{Theorem:robustness_and_recovery} one might want to
create measurement matrices with good robustness constants.
Such matrices are near restricted isometries.
\begin{Theorem}[Measurement Matrix with bounded $\kappa$]
	\label{Theorem:optimal_matrix_kappa_bounded}
	Let $C\subset F\subset\mathbb{K}^N$,
	$\normLHS{\cdot}$ a norm on $\mathbb{K}^N$,
	$\normRHS{\cdot}$ a norm on $\mathbb{K}^M$
	and $\AMat\in\mathbb{K}^{M\times N}$.
	If $\tau\left(\AMat\right)>0$, then the following holds true:
	For all $\zVec\in F$ and $\xVec\in C$ one has
	\begin{align}
		\label{Equation:Theorem:optimal_matrix_kappa_bounded:inequality}
			\tau\left(\AMat\right)
			\normLHS{\zVec-\xVec}
		\leq
			\normRHS{\AMat\left(\zVec-\xVec\right)}
		\leq
			\kappa\left(\AMat\right)
			\normLHS{\zVec-\xVec}
	\end{align}
	In particular, one can find a measurement matrix with a recovery guarantee by finding a maximizer of
	\begin{align}
		\label{Equation:Theorem:optimal_matrix_kappa_bounded:optimization}
		\max_{
			\AMat'\in\mathbb{K}^{M\times N}:
			\kappa\left(\AMat'\right)\leq 1
		}
		\tau\left(\AMat'\right).
	\end{align}
\end{Theorem}
\begin{proof}
	If $\zVec=\xVec$, then also $\AMat\zVec=\AMat\xVec$ and both
	bounds are fulfilled immediately. Thus, assume $\zVec\neq \xVec$.
	By definition of $\tau\left(\AMat\right)$ and $\kappa\left(\AMat\right)$
	one gets
	\begin{align}
		\nonumber
			\normRHS{\AMat\left(\zVec-\xVec\right)}
		=
			\normRHS{\AMat\frac{\zVec-\xVec}{\normLHS{\zVec-\xVec}}}
			\normLHS{\zVec-\xVec}
		\geq
			\tau\left(\AMat\right)\normLHS{\zVec-\xVec},
	\end{align}
	and
	\begin{align}
		\nonumber
			\normRHS{\AMat\left(\zVec-\xVec\right)}
		=
			\normRHS{\AMat\frac{\zVec-\xVec}{\normLHS{\zVec-\xVec}}}
			\normLHS{\zVec-\xVec}
		\leq
			\kappa\left(\AMat\right)\normLHS{\zVec-\xVec},
	\end{align}
	which finishes the proof.
\end{proof}
If one can evaluate the function $\tau$, one can potentially solve \eqref{Equation:Theorem:optimal_matrix_kappa_bounded:optimization}
in order to generate a good measurement matrix deterministically.
The matrix from this optimization problem is making the inequality in
\refP{Equation:Theorem:optimal_matrix_kappa_bounded:inequality}
as tight as possible creating a near restricted isometry.
Since $\kappa$ is seminorm, the constraint $\kappa\left(\AMat'\right)\leq 1$ creates
a convex set of feasible points. However, the function $\tau$ is not necessarily convex
and so the optimization problem is not convex either.
If the vector is not in $C$ but $\xVec\in F$, one might want
to bound the estimation error rather by a linear function in the distance
of $\xVec$ to $C$.
\begin{Theorem}[Stability as a consequence of Robustness]
	\label{Theorem:robustness=>stability}
	Let $C\subset F\subset\mathbb{K}^N$,
	$\normLHS{\cdot}$ a norm on $\mathbb{K}^N$,
	$\normRHS{\cdot}$ a norm on $\mathbb{K}^M$
	and $\AMat\in\mathbb{K}^{M\times N}$.
	If $\tau\left(\AMat\right)>0$, then the following holds true:\\
	For all $\xVec\in F$, $\yVec\in\mathbb{K}^M$
	and $\xVec^\#\in\argmin{\zVec\in F}\normRHS{\AMat\zVec-\yVec}$
	one has
	\begin{align}
		\label{Equation:robustness=>stability:stability}
		\normLHS{\xVec^\#-\xVec}
		\leq\left(1+\frac{\kappa\left(\AMat\right)}{\tau\left(\AMat\right)}\right)
				\inf_{\zVec\in C}\normLHS{\zVec-\xVec}
			+\frac{2}{\tau\left(\AMat\right)}\normRHS{\yVec-\AMat\xVec}
	\end{align}
\end{Theorem}
\begin{proof}
	Let $\left(\xVec^k\right)_{k\in\mathbb{N}}\subset C$ be a sequence such that
	$\lim_{k\rightarrow\infty}\normLHS{\xVec^k-\xVec}
		=\inf_{\zVec\in C}\normLHS{\zVec-\xVec}$.
	Using \eqref{Equation:Theorem:optimal_matrix_kappa_bounded:inequality}
	with $\xVec^\#\in F$ and $\xVec^k\in C$ gives
	\begin{align}
		\nonumber
			\normLHS{\xVec^\#-\xVec}
		\leq&
			\normLHS{\xVec^k-\xVec}+\normLHS{\xVec^\#-\xVec^k}
		\overset{\eqref{Equation:Theorem:optimal_matrix_kappa_bounded:inequality}}{\leq}
			\normLHS{\xVec^k-\xVec}
			+\frac{1}{\tau\left(\AMat\right)}\normRHS{\AMat\xVec^\#-\AMat\xVec^k}
		\\\leq&\nonumber
			\normLHS{\xVec^k-\xVec}
			+\frac{1}{\tau\left(\AMat\right)}\normRHS{\AMat\xVec^\#-\yVec}
			+\frac{1}{\tau\left(\AMat\right)}\normRHS{\yVec-\AMat\xVec}
			+\frac{1}{\tau\left(\AMat\right)}\normRHS{\AMat\xVec-\AMat\xVec^k}
		\\\leq&\nonumber
			\left(1+\frac{\kappa\left(\AMat\right)}{\tau\left(\AMat\right)}\right)
			\normLHS{\xVec^k-\xVec}
			+\frac{1}{\tau\left(\AMat\right)}\normRHS{\AMat\xVec^\#-\yVec}
			+\frac{1}{\tau\left(\AMat\right)}\normRHS{\yVec-\AMat\xVec}.
	\end{align}
	Applying that $\xVec^\#$ is a minimizer of
	$\min_{\zVec\in F}\normRHS{\AMat\zVec-\yVec}$ and
	$\xVec\in F$ is feasible to this gives
	\begin{align}
		\nonumber
			\normLHS{\xVec^\#-\xVec}
		\leq&
			\left(1+\frac{\kappa\left(\AMat\right)}{\tau\left(\AMat\right)}\right)
			\normLHS{\xVec^k-\xVec}
			+\frac{2}{\tau\left(\AMat\right)}\normRHS{\yVec-\AMat\xVec}.
	\end{align}
	Taking the limit for $k\rightarrow\infty$ yields the claim.
\end{proof}
If the error is bounded linearly in
$\inf_{\zVec\in C}\normLHS{\zVec-\xVec}$,
this is called stability here. In the case $\mathbb{K}=\mathbb{R}$, $\normLHS{\cdot}=\norm{\cdot}_1$,
$C=\Sigma_S^N\cap\mathbb{R}_+^N$ and $F=\mathbb{R}_+^N$
this can be compared to \cite[Chapter~4]{IntroductionCS}.
Note that the stability in this case is a consequence of the robustness.
Moreover, the near restricted isometry from
\eqref{Equation:Theorem:optimal_matrix_kappa_bounded:optimization}
not only has the optimal robustness constant up to a factor
of $2$, it also minimizes the constants in front of both error terms in
\eqref{Equation:robustness=>stability:stability}.
\refP{Statement:Theorem:robustness_and_recovery:tau_robust} from
\thref{Theorem:robustness_and_recovery} shows the optimality of the robustness constant in \eqref{Equation:robustness=>stability:stability}
up to a factor of $2$. A similar statement for the stability constant
is missing. So potentially there could be measurement matrices
with significantly better stability constants than
$1+\frac{\kappa\left(\AMat\right)}{\tau\left(\AMat\right)}$.
The solution of \eqref{Equation:Theorem:optimal_matrix_kappa_bounded:optimization} gives
an insight how robust recovery can be given some dimensions and thus receives a definition here.
\begin{Definition}
	Let $C\subset F\subset\mathbb{K}^N$,
	$\normLHS{\cdot}$ a norm on $\mathbb{K}^N$,
	$\normRHS{\cdot}$ a norm on $\mathbb{K}^M$.
	Then
	\begin{align}
		\nonumber
			\tau_{\kappa}\left(C,F,\mathbb{K},\normLHS{\cdot},\normRHS{\cdot}\right)
		:=
			\sup_{\AMat'\in\mathbb{K}^{M\times N}:\kappa\left(\AMat'\right)\leq 1}
			\tau\left(\AMat'\right)
	\end{align}
	is called complexity constant.
\end{Definition}
Note that due to \thref{Theorem:robustness_and_recovery}
uniform and robust recovery is possible if and only if
the complexity constant is non-zero.
In light of \thref{Theorem:robustness_and_recovery}
it is of interest to check whether $\Cone{F-C}$ is closed.
\begin{Proposition}
	\label{Proposition:cone_closed}
	Let $M,N,S\in\mathbb{N}$. Let one of the following conditions hold true:
	\begin{enumerate}
		\item
			\label{Condition:Proposition:cone_closed:KR}
			$C=\Sigma_S^N$ and $F=\Sigma_S^N$ (the Kruskal Rank case).
		\item
			\label{Condition:Proposition:cone_closed:SKC}
			$C=\Sigma_S^N\cap\mathbb{R}^N_+$ and $F=\mathbb{R}^N_+$
			(the Signed Kernel Condition case).
		\item
			\label{Condition:Proposition:cone_closed:NSP}
			$C=\eta\left(\Sigma_S^N\cap\mathbb{S}^{N-1}_{\norm{\cdot}_1}\right)$
			and $F=\eta\mathbb{B}^{N}_{\norm{\cdot}_1}$,
			where $\eta>0$ (the Null Space Property case).
		\item
			\label{Condition:Proposition:cone_closed:ONP}
			$C=\eta\left(\Sigma_S^N\cap\mathbb{S}^{N-1}_{\norm{\cdot}_1}\cap\mathbb{R}_+^N\right)$
			and $F=\eta\left(\mathbb{B}^{N}_{\norm{\cdot}_1}\cap\mathbb{R}_+^N\right)$,
			where $\eta>0$ (Outwardly Neighborly Polytope case).
		\item
			\label{Condition:Proposition:cone_closed:FCB}
			$\SetSize{C}<\infty$ and $F=\conv{C}$ (the Finite Code Book Case).
	\end{enumerate}
	Then, $\Cone{F-C}$ is closed.
\end{Proposition}
\begin{proof}
	Let $A=\left\{\aVec^1,\dots,\aVec^K\right\}\subset\mathbb{K}^N$
	and consider the matrix $\AMat$ whose columns are all
	$\aVec^k$. One gets
	\begin{align}
		\nonumber
			\Cone{\conv{A}}
		=&
			\left\{
				\sum_{k=1}^K\alpha v_k\aVec^k:
				\sum_{k=1}^K v_k=1,
				v_k\geq 0\TextForAll k\in\SetOf{K},
				\alpha\geq 0
			\right\}
		\\=&\nonumber
			\left\{
				\sum_{k=1}^K\aVec^k\vVec_k:
				\vVec\in\mathbb{R}_+^N
			\right\}
		=
			\AMat\mathbb{R}_+^N
	\end{align}
	and the right hand side is closed in $\mathbb{K}^N$. Thus, the following logical
	statement is true:
	\begin{align}
		\label{Equation:Proposition:cone_closed:finite}
		\textnormal{
			If $A$ is finite, then $\Cone{\conv{A}}$ is closed.
		}
	\end{align}
	\par
	Assume \refP{Condition:Proposition:cone_closed:KR} is fulfilled.
	Note that $F-C=\Sigma_{2S}^N$ which is a closed cone.
	Thus $\Cone{F-C}$ is a closed cone.
	\par
	Assume \refP{Condition:Proposition:cone_closed:SKC} is fulfilled.
	Let $\left(\vVec^k\right)_{k\in\mathbb{N}}\subset\Cone{F-C}$
	be convergent to some $\vVec$.
	Define $T_k:=\left\{n\in\SetOf{N}:\vVec^k_n<0\right\}\subset\SetOf{N}$
	so that $\SetSize{T_k}\leq S$.
	Since $\SetOf{N}$ is finite, there is a $T\subset\SetOf{N}$ with $\SetSize{T}\leq S$
	such that $T_k=T$ holds for infinitely many $k$.
	Hence, by replacing $\left(\vVec^k\right)_{k\in\mathbb{N}}$ by one of its subsequences,
	it can be and is assumed that $T_k=T$ for all $k\in\mathbb{N}$.
	Define $\zVec^k:=\ProjToIndex{T^c}{\vVec^k}$
	and $\xVec^k:=-\ProjToIndex{T}{\vVec^k}$.
	Due to the definition of $T_k$ it follows that $\zVec^k=\ProjToIndex{T_k^c}{\vVec^k}\in F$ and
	$\xVec^k=-\ProjToIndex{T_k}{\vVec^k}\in C$.
	Since $\left(\vVec^k\right)_{k\in\mathbb{N}}$ is converging to $\vVec$,
	it follows that
	$\lim_{k\rightarrow\infty}\zVec^k=\lim_{k\rightarrow\infty}\ProjToIndex{T^c}{\vVec^k}
		=\ProjToIndex{T^c}{\vVec}$
	and
	$\lim_{k\rightarrow\infty}\xVec^k=\lim_{k\rightarrow\infty}-\ProjToIndex{T}{\vVec^k}
		=-\ProjToIndex{T}{\vVec}$.
	Due to $C,F$ being closed it follows that $\zVec:=\ProjToIndex{T^c}{\vVec}\in F$ and
	$\xVec:=-\ProjToIndex{T}{\vVec}\in C$.
	This yields that
	$\vVec=\lim_{k\rightarrow\infty}\vVec^k=\lim_{k\rightarrow\infty}\zVec^k-\xVec^k=\zVec-\xVec
		\in F-C\subset\Cone{F-C}$.
	Thus $\Cone{F-C}$ is closed.
	\par
	Assume \refP{Condition:Proposition:cone_closed:NSP} is fulfilled.
	The cases $\mathbb{K}=\mathbb{R}$ and $\mathbb{K}=\mathbb{C}$
	are fundamentally different. This is because $\ell_1$-ball in $\mathbb{R}^N$ is a convex polytope,
	but the $\ell_1$-ball in $\mathbb{C}^N$ is not a convex polytope due to the
	$\ell_2$ nature of the absolute value. Thus, these two cases need to be considered independently.
	\par
	Assume $\mathbb{K}=\mathbb{R}$.	Given $T\subset\SetOf{N}$ with $\SetSize{T}=S$ set
	\begin{align}
		\nonumber
			C_T
		:=
			\left\{
				\xVec\in \eta\mathbb{S}^{N-1}_{\norm{\cdot}_1}:
				\supp{\xVec}\subset T
			\right\}
	\end{align}
	so that
	\begin{align}
		\label{Equation:Proposition:cone_closed:nsp:CT}
			C
		=
			\bigcup_{T\subset\SetOf{N}:\SetSize{T}=S}C_T.
	\end{align}
	Each $C_T$ is a real $S-1$-dimensional $\ell_1$-unit sphere embedded
	in $\mathbb{R}^N$ which is the boundary of a polytope.
	Hence, $C_T$ is the union of finitely many faces.
	To be precise there exists
	faces $f\left(T,k\right)\subset\mathbb{R}^N$
	for $k\in\SetOf{K(T)}$ such that
	\begin{align}
		\label{Equation:Proposition:cone_closed:nsp:fTk}
			C_T
		=
			\bigcup_{k=1}^{K(T)}f\left(T,k\right).
	\end{align}
	Each face $f\left(T,k\right)$ is a closed set and the convex
	hull of finitely many points. Thus, there exist
	a finite set $C\left(T,k\right)$ such that
	\begin{align}
		\label{Equation:Proposition:cone_closed:nsp:CTk}
			f\left(T,k\right)
		=
			\conv{C\left(T,k\right)}.
	\end{align}
	Combining \eqref{Equation:Proposition:cone_closed:nsp:CT}, 
	\eqref{Equation:Proposition:cone_closed:nsp:fTk} and
	\eqref{Equation:Proposition:cone_closed:nsp:CTk}
	yields
	\begin{align}
		\label{Equation:Proposition:cone_closed:nsp:C}
			C
		=
			\bigcup_{T\subset\SetOf{N}:\SetSize{T}=S}C_T
		=
			\bigcup_{T\subset\SetOf{N}:\SetSize{T}=S}\bigcup_{k=1}^{K(T)}
			f\left(T,k\right)
		=
			\bigcup_{T\subset\SetOf{N}:\SetSize{T}=S}\bigcup_{k=1}^{K(T)}
			\conv{C\left(T,k\right)}.
	\end{align}
	Let $f$ be the set of all standard unit vectors and their negatives,
	then $F=\conv{f}$. Using this and
	\eqref{Equation:Proposition:cone_closed:nsp:C} yields
	\begin{align}
		\nonumber
			F-C
		=
			\conv{f}-
			\bigcup_{T\subset\SetOf{N}:\SetSize{T}=S}\bigcup_{k=1}^{K(T)}
			\conv{f\left(C_T,k\right)}
		=
			\bigcup_{T\subset\SetOf{N}:\SetSize{T}=S}\bigcup_{k=1}^{K(T)}
			\left(
				\conv{f}-\conv{C\left(T,k\right)}
			\right).
	\end{align}
	Applying \cite[Theorem~1.1.2]{convex} to this gives
	\begin{align}
		\label{Equation:Proposition:cone_closed:nsp:cone(F-C)}
			\Cone{F-C}
		=
			\Cone{
				\bigcup_{T\subset\SetOf{N}:\SetSize{T}=S}\bigcup_{k=1}^{K(T)}
				\conv{f-C\left(T,k\right)}
			}
		=
			\bigcup_{T\subset\SetOf{N}:\SetSize{T}=S}\bigcup_{k=1}^{K(T)}
			\Cone{\conv{f-C\left(T,k\right)}}.
	\end{align}
	Since all $C\left(T,k\right)$ and $f$ are finite, also all
	$f-C\left(T,k\right)$ are finite. By 
	\eqref{Equation:Proposition:cone_closed:finite}
	all $\Cone{\conv{f-C\left(T,k\right)}}$ are closed
	and by \eqref{Equation:Proposition:cone_closed:nsp:cone(F-C)}
	$\Cone{F-C}$ is a closed set as a finite union of closed sets.
	\par
	Now let $\mathbb{K}=\mathbb{C}$ and note that the claim is proven for the real case,
	i.e. it is already known that $\Cone{F\cap\mathbb{R}^N-C\cap\mathbb{R}^N}$ is closed.
	Define the function $g:\mathbb{C}^N\rightarrow\mathbb{R}_+^N$ mapping $\xVec\in \mathbb{C}^N$
	to $\left(\abs{\xVec_n}\right)_{n\in\SetOf{N}}$, i.e. mapping each component to its absolute value.
	Let $\zVec\in F,\xVec\in C$. Define $\xVec',\zVec'\in\mathbb{R}^N$ by $\xVec'_n:=-\abs{\xVec_n}$
	and $\zVec'_n:=\abs{\zVec_n-\xVec_n}-\abs{\xVec_n}$.
	On the one hand one gets
	\begin{align}
		\nonumber
			\zVec'_n
		\leq
			\abs{\zVec_n}+\abs{\xVec_n}-\abs{\xVec_n}
		=
			\abs{\zVec_n}
	\end{align}
	and on the other hand it follows that
	\begin{align}
		\nonumber
			\zVec'_n
		=
			\abs{\zVec_n-\xVec_n}-\abs{\xVec_n-\zVec_n+\zVec_n}
		\geq
			\abs{\zVec_n-\xVec_n}-\abs{\xVec_n-\zVec_n}-\abs{\zVec_n}
		=
			-\abs{\zVec_n}	.
	\end{align}
	This yields $\norm{\zVec'}_1\leq\norm{\zVec}_1$. Since $\zVec\in F$, it follows that
	$\zVec'\in F\cap\mathbb{R}^N$.
	Due to $\xVec\in C$, one directly gets $\xVec'\in C\cap\mathbb{R}^N$.
	Lastly, $\zVec'_n-\xVec'_n=\abs{\zVec_n-\xVec_n}\geq 0$ which together with the previous two
	statements yields
	$g\left(\zVec-\xVec\right)=\zVec'-\xVec'\in
		\left(F\cap\mathbb{R}^N-C\cap\mathbb{R}^N\right)\cap\mathbb{R}_+^N$.
	It follows that
	\begin{align}
		\label{Equation:Proposition:cone_closed:nsp:preimage_subset}
		F-C\subset g^{-1}\left(\left(F\cap\mathbb{R}^N-C\cap\mathbb{R}^N\right)\cap\mathbb{R}_+^N\right).
	\end{align}
	Now let $\vVec\in g^{-1}\left(\left(F\cap\mathbb{R}^N-C\cap\mathbb{R}^N\right)\cap\mathbb{R}_+^N\right)$.
	Then, there exists $\zVec\in F\cap\mathbb{R}^N,\xVec\in C\cap\mathbb{R}^N$ such that
	$g\left(\vVec\right)=\zVec-\xVec$.
	Define the vectors $\zVec',\xVec'\in\mathbb{C}^N$ by
	\begin{align}
		\nonumber
			\zVec'_n
		:=
			\begin{Bmatrix}
					\zVec_n &\TextIf \vVec_n=0
				\\
					\frac{\vVec_n}{\abs{\vVec_n}}\zVec_n &\TextIf \vVec_n\neq 0
			\end{Bmatrix}
		\TextAnd
			\xVec'_n
		:=
			\begin{Bmatrix}
					\xVec_n &\TextIf \vVec_n=0
				\\
					\frac{\vVec_n}{\abs{\vVec_n}}\xVec_n &\TextIf \vVec_n\neq 0
			\end{Bmatrix}.
	\end{align}
	It follows that $\norm{\zVec'}_1=\norm{\zVec}_1$ and $\norm{\xVec'}_1=\norm{\xVec}_1$
	and thus $\zVec'\in F$ and $\xVec'\in C$.
	Further, if $\vVec_n=0$, then 
	\begin{align}
		\nonumber
			\zVec'_n-\xVec'_n
		=
			\zVec_n-\xVec_n
		=
			\abs{\vVec_n}
		=
			0
		=
			\vVec_n.
	\end{align}
	On the other if $\vVec_n\neq 0$, then
	\begin{align}
		\nonumber
			\zVec'_n-\xVec'_n
		=
			\frac{\vVec_n}{\abs{\vVec_n}}\left(\zVec_n-\xVec_n\right)
		=
			\frac{\vVec_n}{\abs{\vVec_n}}\abs{\vVec_n}
		=
			\vVec_n.
	\end{align}
	Thus $\vVec\in F-C$. Together with \eqref{Equation:Proposition:cone_closed:nsp:preimage_subset}
	this yields
	\begin{align}
		\nonumber
			F-C
		=
			g^{-1}\left(\left(F\cap\mathbb{R}^N-C\cap\mathbb{R}^N\right)\cap\mathbb{R}_+^N\right).
	\end{align}
	Applying the homogeneity of $g$ yields
	\begin{align}
		\label{Equation:Proposition:cone_closed:nsp:preimage_cone}
			\Cone{F-C}
		=&
			g^{-1}\left(\Cone{\left(F\cap\mathbb{R}^N-C\cap\mathbb{R}^N\right)\cap\mathbb{R}_+^N}\right).
	\end{align}
	On the other hand one can show that
	\begin{align}
		\label{Equation:Proposition:cone_closed:nsp:cone_shenanigans}
			\Cone{\left(F\cap\mathbb{R}^N-C\cap\mathbb{R}^N\right)\cap\mathbb{R}_+^N}
		=
			\Cone{F\cap\mathbb{R}^N-C\cap\mathbb{R}^N}\cap\mathbb{R}_+^N
	\end{align}
	holds true.
	Due to the real case proven before
	$\Cone{F\cap\mathbb{R}^N-C\cap\mathbb{R}^N}$ is closed in $\mathbb{R}^N$
	and thus also closed in $\mathbb{C}^N$.
	By \eqref{Equation:Proposition:cone_closed:nsp:cone_shenanigans} also
	\begin{align}
		\nonumber
			\Cone{\left(F\cap\mathbb{R}^N-C\cap\mathbb{R}^N\right)
		\cap
			\mathbb{R}_+^N}
	\end{align}
	is closed.
	By \eqref{Equation:Proposition:cone_closed:nsp:preimage_cone} $\Cone{F-C}$ is
	the preimage of a closed set under a continuous function and thus closed.
	\par
	Assume \refP{Condition:Proposition:cone_closed:ONP} is fulfilled.
	Set $C':=\eta\left(\Sigma_S^N\cap\mathbb{S}^{N-1}_{\norm{\cdot}_1}\right)\cap\mathbb{R}^N$
	and $F:=\eta\mathbb{B}^{N}_{\norm{\cdot}_1}\cap\mathbb{R}^N$
	from the real null space property case.
	It follows that
	$C=C'\cap\mathbb{R}^N_+$ and $F=F'\cap\mathbb{R}^N_+$.
	Due to this the proof of the outward neighborly polytope case for any $\mathbb{K}$ proceeds similar
	to the proof of the real null space property case.
	Again one shows that $F-C$ is the union of finitely many convex hulls
	of sets with finitely many elements.
	By \eqref{Equation:Proposition:cone_closed:nsp:C} the cones of these
	convex hulls are closed
	and $\Cone{F-C}$ is closed as the finite union of closed sets.
	The only difference is that there are less faces to consider per embedded $\ell_1$-unit sphere
	due to the intersection with $\mathbb{R}_+^N$.
	A more detailed proof is omitted to shorten this work.
\end{proof}
	\section{The Signed Kernel Condition Case}
In this section the case $C=\Sigma_S^N\cap\mathbb{R}^N_+$ and $F=\mathbb{R}^N_+$ is investigated.
This is a new case and has been studied with $\mathbb{K}=\mathbb{R}$ and $\normRHS{\cdot}=\norm{\cdot}_2$
in \cite{NNLS} and with $\mathbb{K}=\mathbb{R}$ for general $\normRHS{\cdot}$ in \cite{NNLAD}.
In particular it was shown that under the assumption of a null space property
and a minor further condition, recovery is possible
if $M\geq CS\Ln{\ExpE\frac{N}{S}}$.
Here $\mathbb{K}$ can either be the real or the complex numbers.
and it is shown that less requirements and measurements
are sufficient for a recovery guarantee.
Note that since $\Cone{F-C}$ is closed in this case, recovery and robust recovery are equivalent by \thref{Theorem:robustness_and_recovery}.
\begin{Definition}
	Let $\AMat\in\mathbb{K}^{M\times N}$ and $S\in\mathbb{N}$.
	Suppose that
	\begin{align}
		\nonumber
			\SetSize{\left\{n\in \SetOf{N}:\vVec_n<0\right\}}
		>
			S
		\TextForAll
			\vVec\in\Kernel{\AMat}\cap\mathbb{R}^N\setminus\ZeroSet
	\end{align}
	holds true.
	Then, $\AMat$ has the signed kernel condition of order $S$.
\end{Definition}
This condition means that real vectors in the kernel of $\AMat$ apart from
the all zeros vector have a minimal number of negative and also positive
entries since the negative vector is also in the kernel.
It is the equivalent condition for uniform recovery.
\begin{Theorem}\label{Theorem:signed_kernel_condition}
	Let $\normRHS{\cdot}$ be any norm on $\mathbb{K}^M$.
	Let $S\in\mathbb{N}$ and $\AMat\in\mathbb{K}^{M\times N}$.
	Then, the following conditions are equivalent:
	\begin{enumerate}
		\item
			$\AMat$ has the signed kernel condition of order $S$.
		\item
			For all $\xVec\in\Sigma_S\cap\mathbb{R}^N_+$ one has
			$
				\left\{\xVec\right\}=\argmin{\zVec\in\mathbb{R}_+^N}\normRHS{\AMat\zVec-\AMat\xVec}
			$.
	\end{enumerate}
\end{Theorem}
\begin{proof}
	$(1)\Rightarrow (2)$: Let $\vVec\in\left(F-C\right)\cap\Kernel{\AMat}$.
	Since $F,C\subset\mathbb{R}^N$, one has $\vVec\in\Kernel{\AMat}\cap\mathbb{R}^N$.
	Set $T:=\left\{n\in \SetOf{N}:\vVec_n<0\right\}$. By definition of $C$ and $F$ it follows that
	$\SetSize{T}\leq S$. Due to the assumption $\vVec=0$ and thus
	$\left(F-C\right)\cap\Kernel{\AMat}=\ZeroSet$.
	Applying \refP{Statement:Theorem:robustness_and_recovery:kernel_recovery} of
	\thref{Theorem:robustness_and_recovery} yields the claim.
	\par
	$(2)\Rightarrow (1)$:
	Let $\vVec\in\Kernel{\AMat}\cap\mathbb{R}^N$ with
	$\SetSize{T}\leq S$ where $T:=\left\{n\in \SetOf{N}:\vVec_n<0\right\}$.
	Set $\zVec:=\ProjToIndex{T^c}{\vVec}$ and $\xVec:=-\ProjToIndex{T}{\vVec}$.
	It follows that $\zVec\in F$ and $\xVec\in C$ as well as $\vVec=\zVec-\xVec$.
	Thus, $\vVec\in\left(F-C\right)\cap\Kernel{\AMat}$.
	Due to \refP{Statement:Theorem:robustness_and_recovery:kernel_recovery} of
	\thref{Theorem:robustness_and_recovery} one gets
	$\left(F-C\right)\cap\Kernel{\AMat}=\ZeroSet$ and hence $\vVec=0$.
	Thus, any $\vVec\in\Kernel{\AMat}\cap\mathbb{R}^N\setminus\ZeroSet$
	needs to obey $\SetSize{\left\{n\in \SetOf{N}:\vVec_n<0\right\}}> S$.
\end{proof}
In order to show the existence of matrices with a signed kernel condition the following complex
extension of \emph{Descartes' rule of signs} from \cite{complex_rule_of_signs} is required.
It is simplified by ignoring multiplicity of roots.
\begin{Theorem}[Adaption of Theorem 2.4 and Equation (1.1) from \cite{complex_rule_of_signs}]
\label{Theorem:complex_rule_of_signs}
	Let $N'\in\mathbb{N}\cup\ZeroSet$, $\pVec\in\mathbb{R}^N$ and
	$\p{a}=\sum_{n=1}^{N}\pVec_na^{n-1+N'}$ be the corresponding polynomial
	with real coefficients. Let $N''$ be the degree of $\p{}$ and assume $N''\geq 1$.
	Let the following conditions hold true:
	\begin{enumerate}
		\item
			\label{Condition:Theorem:complex_rule_of_signs:positivity}
			Let $\aVec_m>0$ for all $m\in\SetOf{M}$.
		\item
			\label{Condition:Theorem:complex_rule_of_signs:arguments}
			Let the arguments suffice
			$
					\tVec_m
				\in
					\left(-\frac{\pi}{N''+2-M},\frac{\pi}{N''+2-M}\right)
				\TextForAll
					m\in\SetOf{M}$.
		\item
			\label{Condition:Theorem:complex_rule_of_signs:distinctness}
			Let $\Exp{i\tVec_m}\aVec_m$ for all $m\in\SetOf{M}$ be distinct from each other.
		\item
			\label{Condition:Theorem:complex_rule_of_signs:roots}
			Let $\p{\aVec_m\Exp{i\tVec_m}}=0$ for all $m\in\SetOf{M}$.
	\end{enumerate}
	Then, the set of sign variations
	\begin{align}
		\nonumber
			\textnormal{var}\left(\pVec\right)
		:=
			\left\{
				\left(n,n'\right)\in\SetOf{N}^2:
				n<n',\pVec_n\pVec_{n'}<0 \TextAnd \pVec_{n''}=0\TextForAll n'' \TextWith n<n''<n'
			\right\}
	\end{align}
	suffices $\SetSize{\textnormal{var}\left(\pVec\right)}\geq M$.
\end{Theorem}
Note that due to \cite[Equation~(1.1)]{complex_rule_of_signs} the result \cite[Theorem~2.4]{complex_rule_of_signs} has the additional requirement $N''\geq 2$.
It is clear that a polynomial of degree $N''=1$ with one positive real root has at least
one sign change in its coefficient sequence, thus the statement also holds for $N''=1$.
The formulation with $N'$ in the exponents looks odd, but is of advantage for this work
By increasing $N'$ one can artificially increase the degree of the polynomial,
thus make the denominator
$N''+2-M$ positive and allow one to define suitable roots.
If $\AMat$ is the matrix for which a signed kernel condition is supposed to be proven
in this section, then from the proofs of the theorems one can see that
$N'$ will be non-zero if and only if the trivial case
$\Kernel{\AMat}\cap\mathbb{R}^N=\ZeroSet$ is fulfilled.
The result \cite[Theorem~2.4]{complex_rule_of_signs} goes back to
\cite{complex_rule_of_signs_french}.
\subsection{Real Signed Kernel Condition Case}
In this subsection the case $\mathbb{K}=\mathbb{R}$ is being considered.
The following statement allows one to create matrices with the signed kernel condition.
\begin{Theorem}[Real Vandermonde Matrices have $M=2S+1$]\label{Theorem:real_signed_kernel:existence}
	Let $\mathbb{K}=\mathbb{R}$, $\aVec_1,\dots,\aVec_M>0$ be distinct from each other
	and $\AMat\in\mathbb{R}^{M\times N}$
	have entries $\AMat_{m,n}=\aVec_m^{n-1}$.
	Then, $\AMat$ has the signed kernel condition of order
	$\lceil\frac{M}{2}\rceil-1$.
	In particular, if $M=2S+1$, then $\AMat$ has signed kernel
	condition of order $S$.
\end{Theorem}
\begin{proof}
	Let $\pVec\in\Kernel{\AMat}\cap\mathbb{R}^N\setminus\ZeroSet$. Consider the polynomial
	$\p{a}:=\sum_{n=1}^{N}\pVec_{n}a^{n-1}$ which has real coefficients.
	Let $N''$ be the degree of $\p{}$ and note that $N''\leq N-1$.
	Since $0=\AMat\pVec=\left(\p{\aVec_m}\right)_{m\in\SetOf{M}}$
	the polynomial $\p{}$ has the $M$ distinct roots $\aVec_m$.
	If $N''<M$, then $\p{}$ would have more roots than its degree
	and hence would need to be the zero polynomial.
	This would be a contradiction and one would get
	$\Kernel{\AMat}\cap\mathbb{R}^N=\left\{0\right\}$ and that the signed kernel
	condition of any order would be fulfilled. Thus, assume that $N''\geq M$.
	It follows that $M+1\leq N$.
	Due to $\pVec\neq 0$ the polynomial $p$ is not the zero polynomial and $N''\geq 1$.
	By \emph{Descartes' rule of signs} \cite{rule_of_signs}\cite[Theorem~2.2]{complex_rule_of_signs}
	the set $\textnormal{var}\left(\pVec\right)$ contains at least $M$ elements.
	That means that $\pVec$ contains a subsequence of $M+1$ elements
	so that subsequent elements are non-zero and have opposite sign.
	This subsequence and thus $\pVec$ need to have at least
	$\lfloor\frac{M+1}{2}\rfloor=\lceil\frac{M}{2}\rceil$
	negative entries. Hence, $\AMat$ has signed kernel condition of
	order $\lceil\frac{M}{2}\rceil-1$.
\end{proof}
The disadvantage of this result is that the absolute values of coefficients of $\AMat$
grows exponentially in $N$.
Instead one can consider the following trigonometric construction.
\begin{Theorem}[Trigonometric Matrices have $M=2S+1$]\label{Theorem:real_signed_kernel:existence_trigonometric}
	Let $\mathbb{K}=\mathbb{R}$, $S,N\in\mathbb{N}$ and set $M=2S+1$ and $N':=\max\{2S+1-N,0\}$.
	Let the following conditions hold true:
	\begin{enumerate}
		\item
			\label{Condition:Theorem:real_signed_kernel:existence_trigonometric:arguments}
			Let the arguments suffice
			$
					\tVec_s
				\in
					\left(0,\frac{\pi}{N+N'-2S}\right)
				\TextForAll
					s\in\SetOf{S}$.
		\item
			\label{Condition:Theorem:real_signed_kernel:existence_trigonometric:distinctness}
			Let $\tVec_1<\dots<\tVec_S$ and set $\tVec_0:=0$.
	\end{enumerate}
	Let $\AMat\in\mathbb{R}^{M\times N}$ have entries
	$\AMat_{2s+1,n}=\cos\left(\tVec_s\left(n-1+N'\right)\right)$ for all $s\in\SetOf{S}\cup\ZeroSet$ and $n\in\SetOf{N}$
	and $\AMat_{2s,n}=\sin\left(\tVec_s\left(n-1+N'\right)\right)$ for all $s\in\SetOf{S}$
	and $n\in\SetOf{N}$.
	Then, $\AMat$ has the signed kernel condition of order $S$.
\end{Theorem}
\begin{proof}
	Let $\pVec\in\Kernel{\AMat}\cap\mathbb{R}^N\setminus\ZeroSet$.
	Let further $\BMat\in\mathbb{C}^{(S+1)\times N}$ with $\BMat_{s,n}=\Exp{i\tVec_{s-1}\left(n-1+N'\right)}$.
	Note that the real and imaginary parts of the rows of rows of $\BMat$ except for the redundant
	zero row are exactly the rows of $\AMat$ respectively.
	Thus $\Kernel{\BMat}\cap\mathbb{R}^N=\Kernel{\AMat}$ and
	$\pVec\in\Kernel{\BMat}\cap\mathbb{R}^N\setminus\ZeroSet$.
	Now consider the polynomial
	$\p{a}:=\sum_{n=1}^{N}\pVec_{n}a^{n-1+N'}$ which has real coefficients.
	Let $N''$ be the degree of $\p{}$.
	Since $0=\BMat\pVec=\left(\p{\Exp{i\tVec_{s-1}}}\right)_{s\in\SetOf{S+1}}$
	the polynomial $\p{}$ has the complex roots $\Exp{i\tVec_s}$ for $s\in\SetOf{S}$
	and the real root $\Exp{i\tVec_0}$. Since $\p{}$ has real coefficients,
	$\p{}$ also has the $S$ complex roots $\Exp{-i\tVec_s}$ for $s\in\SetOf{S}$.
	Due to the assumed \refP{Condition:Theorem:real_signed_kernel:existence_trigonometric:arguments}
	and \refP{Condition:Theorem:real_signed_kernel:existence_trigonometric:distinctness}
	these $2S+1$ roots are all distinct from each other. Thus, these $2S+1$ roots suffice
	\refP{Condition:Theorem:complex_rule_of_signs:distinctness} and
	\refP{Condition:Theorem:complex_rule_of_signs:roots} of
	\thref{Theorem:complex_rule_of_signs}.
	Since $N''\leq N+N'-1$, one gets
	$\tVec_s\in\left(0,\frac{\pi}{N''+2-(2S+1)}\right)$
	for all $s\in\SetOf{S}$.
	Hence all $2S+1$ roots suffice
	\refP{Condition:Theorem:complex_rule_of_signs:positivity} and
	\refP{Condition:Theorem:complex_rule_of_signs:arguments} of
	\thref{Theorem:complex_rule_of_signs}.
	If $N''<2S+1$, then $\p{}$ would have more roots than its degree
	and hence would need to be the zero polynomial. This would be a
	contradiction and one would get
	$\Kernel{\AMat}\cap\mathbb{R}^N=\Kernel{\BMat}\cap\mathbb{R}^N=\left\{0\right\}$
	and that the signed kernel condition of any order would be fulfilled.
	Thus, assume that $N''\geq 2S+1$. It follows that $N''\geq 1$.
	If $2S+1\geq N$, then one gets the contradiction
	$2S+2\leq N''+1\leq N+N'=2S+1$. Hence, $2S+2\leq N$.
	By \thref{Theorem:complex_rule_of_signs} $\SetSize{\textnormal{var}\left(\pVec\right)}\geq 2S+1$.
	That means that $\pVec$ contains a subsequence of $2S+2$ elements
	so that subsequent elements are non-zero and have opposite sign.
	This subsequence and thus $\pVec$ need to have at least
	$S+1$ negative entries. Hence, $\AMat$ has signed kernel condition of
	order $S$.
\end{proof}
Note that a non-uniform robust recovery guarantee with the non-negative least squares
with similarly many measurements has been proven for a trigonometric construction in
\cite[Theorem~3]{nnls_donoho}. However, since the guarantee from \cite[Theorem~3]{nnls_donoho}
is not uniform in all $\xVec\in\Sigma_S^N\cap\mathbb{R}_+^N$
and it is not clear that it can be solved in polynomial time,
this result can not fill in the question mark in \refP{Table:solved_problems}.
Further, any trigonometric construction like the one from
\thref{Theorem:real_signed_kernel:existence_trigonometric}
has the disadvantage that $\AMat$ can not be represented by finitely many bits.
To fill out the missing case in \refP{Table:solved_problems}.
it remains to show that recovery can be performed in a polynomial time.
\begin{Remark}\label{Remark:polynomial_time}
	Let $\mathbb{K}=\mathbb{R}$.If $M\geq 2S+1$, then uniform and robust recovery
	of non-negative $S$-sparse vectors is possible in polynomial time
	in the input bits\footnote{The input bits are the bits used to represent the input, i.e. $\AMat$ and $\yVec$.}.
\end{Remark}
\begin{proof}
	By \thref{Theorem:real_signed_kernel:existence} with $\aVec_m:=\frac{1}{m}$
	there there exists measurement matrices $\AMat\in\mathbb{R}^{M\times N}$
	that can be represented by finitely bits
	with $M=2S+1$ and signed kernel condition of order $S$.
	By \thref{Theorem:signed_kernel_condition} and
	\thref{Theorem:robustness_and_recovery}
	there exists a decoder
	that allow robust recovery of non-negative $S$-sparse vectors. It remains to show
	that the problem can be solved in polynomial time. This may not be
	the case in general, however if
	$\normRHS{\cdot}=\norm{\cdot}_\infty$, one can introduce
	the slack variable $t$ with the constraints
	$\left(\AMat\zVec-\yVec\right)_m\leq t$,
	$-\left(\AMat\zVec-\yVec\right)_m\leq t$ for all $m$
	and minimize $t$ instead of $\norm{\AMat\zVec-\yVec}_\infty$,
	i.e. solve
	$\argmin{\zVec\in\mathbb{R}_+^N,t\in\mathbb{R}:
		\left(\AMat\zVec-\yVec\right)_m\leq t,-\left(\AMat\zVec-\yVec\right)_m\leq t}t$.
	This problem on the other hand is a linear program.
	Linear programs can be solved in polynomial time of the input
	bits \cite{linear_programming} and to transfer the original
	problem into the linear program can also be done in a polynomial time
	of the input bits.
\end{proof}
It should be noted that in simulations the recovery with
matrices from \thref{Theorem:real_signed_kernel:existence} does not appear to be robust.
Here two explanations are provided.
First, the machine precision might simply be insufficient to accurately represent
the exponentially growing real entries of the Vandermonde matrix.
This can be be avoided by considering the construction of
\thref{Theorem:real_signed_kernel:existence_trigonometric}.
Moreover by \thref{Proposition:cone_closed} and \thref{Theorem:openness}
any matrix in $\mathbb{Q}^{M\times N}\subset\mathbb{R}^{M\times N}$ sufficiently close
would also have a signed kernel condition and can be represented by finitely many bits.
The other explanation is that the robustness constant
gets arbitrarily small for $M\in\mathcal{O}\left(S\right)$
and $N\rightarrow \infty$ according to
\thref{Theorem:bound_complexity_constant}.
Due to \refP{Statement:Theorem:robustness_and_recovery:tau_robust} from
\thref{Theorem:robustness_and_recovery} also any constant for a robust recovery
guarantee needs to grow arbitrarily large.
In contrast to the first explanation, the second explanation is proven by
\thref{Theorem:bound_complexity_constant} below.
\par
The previous constructions allow to characterize exactly when the
complexity constant is $0$.
\begin{Theorem}\label{Theorem:signed_kernel_dimensions}
	Let $\mathbb{K}=\mathbb{R}$, $\normLHS{\cdot}$ a norm on $\mathbb{R}^N$ and
	$\normRHS{\cdot}$ a norm on $\mathbb{R}^M$.
	Let $N,M,S\in\mathbb{N}$.
	Then, the following are equivalent:
	\begin{enumerate}
		\item
			$
				\tau_{\kappa}\left(\Sigma_S^N\cap\mathbb{R}_+^N,\mathbb{R}_+^N,\mathbb{R},\normLHS{\cdot},\normRHS{\cdot}\right)>0
			$.
		\item
			$M\geq \min\{2S+1,N\}$.
	\end{enumerate}
\end{Theorem}
\begin{proof}
	$(1)\Rightarrow (2)$:
	There exist a matrix $\AMat$ with $\kappa\left(\AMat\right)\leq 1$ and
	$\tau\left(\AMat\right)
		>\frac{1}{2}\tau_{\kappa}\left(\Sigma_S^N\cap\mathbb{R}_+^N,\mathbb{R}_+^N,\mathbb{R},\normLHS{\cdot},\normRHS{\cdot}\right)>0$.
	Let $\BMat\in\mathbb{R}^{M\times \min\{2S+1,N\}}$
	whose columns are the first $\min\{2S+1,N\}$ columns of $\AMat$.
	By \thref{Theorem:robustness_and_recovery} and
	\thref{Theorem:signed_kernel_condition}
	$\AMat$ has signed kernel condition of order $S$.
	Now let $\vVec\in\Kernel{\AMat}\setminus\ZeroSet$.
	Since $\vVec,-\vVec\in\Kernel{\AMat}\setminus\ZeroSet$,
	the signed kernel condition yields that $\vVec$ needs to have at least
	$S+1$ negative entries and $S+1$ positive entries
	Thus, $\norm{\vVec}_0\geq 2S+2\geq \min\{2S+1,N\}+1$.
	Since this holds for all $\vVec\in\Kernel{\AMat}\setminus\ZeroSet$,
	it follows that $\Kernel{\BMat}=\ZeroSet$.
	However, this is only possible if $M\geq \min\{2S+1,N\}$.
	\par
	$(2)\Rightarrow (1)$:
	To show that the complexity constant is non-zero,
	one matrix with non-zero robustness constant will be constructed.
	If $2S+1\leq N$, then
	\begin{align}
		\nonumber
			\left\lceil\frac{M}{2}\right\rceil-1
		\geq
			\left\lceil\frac{2S+1}{2}\right\rceil-1
		=
			S
	\end{align}
	and the matrix $\AMat$ from \thref{Theorem:real_signed_kernel:existence}
	has the signed kernel condition of order $S$.
	If $N\leq 2S+1$, then $M\geq N$ and there exists
	$\AMat$ with $\Kernel{\AMat}=\ZeroSet$ which immediately has
	the signed kernel condition of order $S$.
	In both cases $\AMat$ has the signed kernel condition of order $S$
	and by \thref{Theorem:signed_kernel_condition} and
	\thref{Theorem:robustness_and_recovery} it follows that
	$\kappa\left(\AMat\right)\geq\tau\left(\AMat\right)>0$.
	By considering the normalized
	$\AMat':=\frac{1}{\kappa\left(\AMat\right)}\AMat$ 
	one gets
	$\tau_{\kappa}\left(\Sigma_S^N\cap\mathbb{R}_+^N,\mathbb{R}_+^N,\mathbb{R},\normLHS{\cdot},\normRHS{\cdot}\right)
		\geq\tau\left(\AMat'\right)
		\geq\frac{\tau\left(\AMat\right)}{\kappa\left(\AMat\right)}>0$.
\end{proof}
Evaluating $\tau\left(\AMat\right)$ would be interesting to determine whether
recovery guarantees are fulfilled and in view of 
\eqref{Equation:Theorem:optimal_matrix_kappa_bounded:optimization}
it could enable methods to determine measurement matrices with recovery guarantees.
Indeed, it can be evaluated in the signed kernel condition case
given certain norms.
\begin{Theorem}\label{Theorem:verifiability}
	Let $\mathbb{K}=\mathbb{R}$, $\AMat\in\mathbb{R}^{M\times N}$, $S\in\mathbb{N}$,
	and $\normRHS{\cdot}$ be a norm on $\mathbb{R}^M$.
	Fix $\normLHS{\cdot}:=\norm{\cdot}_\infty$,
	$C:=\Sigma_S^N\cap\mathbb{R}_+^N$ and $F:=\mathbb{R}_+^N$.
	Given $n\in\SetOf{N}$ and $T\subset\SetOf{N}$ consider
	$pos_n:=\left\{\zVec\in\mathbb{R}^N:\norm{\zVec}_\infty=1,\zVec_n=1\right\}$ and
	$neg_n:=\left\{\zVec\in\mathbb{R}^N:\norm{\zVec}_\infty=1,\zVec_n=-1\right\}$
	and set
	\begin{align}
		\nonumber
			\alpha_1(n,T)
		:=
			\inf_{\vVec\in pos_n:\ProjToIndex{T^c}{\vVec}\in\mathbb{R}_+^N}
			\normRHS{\AMat\vVec}
		\TextAnd
			\alpha_2(n,T)
		:=
			\inf_{\vVec\in neg_n:\ProjToIndex{T^c}{\vVec}\in\mathbb{R}_+^N}
			\normRHS{\AMat\vVec}.
	\end{align}
	Then,
	\begin{align}
		\nonumber
			\tau\left(\AMat\right)
		=
			\inf_{T\subset\SetOf{N},\SetSize{T}=S}
			\min\left\{
				\inf_{n\in\SetOf{N}}\alpha_1(n,T),\inf_{n\in T}\alpha_2(n,T)
			\right\}.
	\end{align}
	In particular, the robustness constant can be calculated by solving
	$\binom{N}{S}\left(N+S\right)$ convex optimization problems
	and, for a suitable choice of $\normRHS{\cdot}$,
	each problem can be solved by a linear program.
\end{Theorem}
\begin{proof}
	Using definition of $C$, $F$ and $\normLHS{\cdot}$ one gets
	\begin{align}
		\nonumber
			\left(F-C\right)\cap\mathbb{S}^{N-1}_{\normLHS{\cdot}}
		=&
			\bigcup_{T\subset\SetOf{N},\SetSize{T}=S}
			\left\{\vVec\in\mathbb{R}^N:\norm{\vVec}_\infty=1\TextAnd \ProjToIndex{T^c}{\vVec}\in\mathbb{R}_+^N\right\}
		\\=&\nonumber
			\bigcup_{T\subset\SetOf{N},\SetSize{T}=S}\bigcup_{n\in\SetOf{N}}
			\left(
				\left\{\vVec\in pos_n:\ProjToIndex{T^c}{\vVec}\in\mathbb{R}_+^N\right\}
				\cup
				\left\{\vVec\in neg_n:\ProjToIndex{T^c}{\vVec}\in\mathbb{R}_+^N\right\}
			\right).
		\\=&\nonumber
			\bigcup_{T\subset\SetOf{N},\SetSize{T}=S}\left(
				\bigcup_{n\in\SetOf{N}}
				\left\{\vVec\in pos_n:\ProjToIndex{T^c}{\vVec}\in\mathbb{R}_+^N\right\}
			\right)\cup\left(
				\bigcup_{n\in T}
				\left\{\vVec\in neg_n:\ProjToIndex{T^c}{\vVec}\in\mathbb{R}_+^N\right\}
			\right),
	\end{align}
	where the last step follows from the fact that
	$\left\{\vVec\in neg_n:\ProjToIndex{T^c}{\vVec}\in\mathbb{R}_+^N\right\}$
	is empty whenever $n\in T^c$.
	Using this and that $F-C$ is a cone the robustness constant is given as
	\begin{align}
		\nonumber
			\tau\left(\AMat\right)
		=&
			\inf_{\vVec\in F-C:\vVec\neq 0}
			\frac{
				\normRHS{\AMat\vVec}
			}{
				\normLHS{\vVec}
			}
		=
			\inf_{\vVec\in\left(F-C\right)\cap\mathbb{S}^{N-1}_{\normLHS{\cdot}}}
			\normRHS{\AMat\vVec}
		\\=&\nonumber
			\inf_{T\subset\SetOf{N},\SetSize{T}=S}
			\min\left\{
				\inf_{n\in\SetOf{N}}\alpha_1\left(n,T\right),
				\inf_{n\in T}\alpha_2\left(n,T\right)
			\right\}.
	\end{align}
	Since $\vVec\mapsto\normRHS{\AMat\vVec}$ is a convex function on the convex set
	$\left\{\vVec\in pos_n:\ProjToIndex{T^c}{\vVec}\in\mathbb{R}_+^N\right\}$,
	$\alpha_1(n,T)$ is a convex optimization
	problem and similarly $\alpha_2(T)$ is a convex optimization problem.
	In total one can calculate
	$\tau\left(\AMat\right)$ by solving
	$\binom{N}{S}\left(N+S\right)$ convex optimization problems.
	The sets $pos_n$ and $neg_n$ are affine sets and can be described
	by linear constraints. By choosing a suitable norm $\normRHS{\cdot}$
	and introducing slack variables similar as in the proof of
	\thref{Remark:polynomial_time}
	each of these problem can be solved by a linear program.
\end{proof}
\subsection{Complex Signed Kernel Condition Case}
In this subsection $\mathbb{K}=\mathbb{C}$ is being considered.
Note that \thref{Remark:polynomial_time} can not be extended to the complex case. This is due to the $\ell_2$ nature of the complex absolute value which prevents the complex
$\ell_p$-norm-minimization from being cast as a linear program.
\begin{Theorem}\label{Theorem:complex_signed_kernel_dimensions}
	Let $\mathbb{K}=\mathbb{C}$, $\normLHS{\cdot}$ a norm on $\mathbb{C}^N$ and
	$\normRHS{\cdot}$ a norm on $\mathbb{C}^M$.
	Let $N,M,S\in\mathbb{N}$.
	Then, the following are equivalent:
	\begin{enumerate}
		\item
			$
				\tau_{\kappa}\left(\Sigma_S^N\cap\mathbb{R}_+^N,\mathbb{R}_+^N,\mathbb{C},\normLHS{\cdot},\normRHS{\cdot}\right)>0
			$.
		\item
			$M\geq \min\{S+1,\frac{1}{2}N\}$.
	\end{enumerate}
\end{Theorem}
\begin{proof}
	The proof follows from the real case.
	To avoid confusion the term real and complex is added
	in front of the signed kernel condition to specify that. Further it should be noted
	that matrices will have kernels over different fields.
	\par
	$(1)\Rightarrow (2)$:
	By assumption there exists a matrix $\BMat\in\mathbb{K}^{M\times N}$ with $\tau\left(\BMat\right)>0$.
	By \refP{Statement:Theorem:robustness_and_recovery:robust=>recovery} of
	\thref{Theorem:robustness_and_recovery} and \thref{Theorem:signed_kernel_condition}
	$\BMat$ has complex signed kernel condition of order $S$.
	Define $\AMat\in\mathbb{R}^{2M\times N}$ by letting the first $M$ rows of $\AMat$ be
	the real part of $\BMat$ and the last $M$ rows of $\AMat$ be the imaginary part of $\BMat$.
	It follows that $\Kernel{\BMat}\cap\mathbb{R}^N=\Kernel{\AMat}=\Kernel{\AMat}\cap\mathbb{R}^N$.
	Thus, $\AMat$ has real signed kernel condition of order $S$.
	Further, by \thref{Proposition:cone_closed} $\Cone{F-C}$ is closed.
	By \refP{Statement:Theorem:robustness_and_recovery:recovery=>robust} of
	\thref{Theorem:robustness_and_recovery} it follows that $\tau\left(\AMat\right)>0$.
	By \thref{Theorem:signed_kernel_dimensions} one gets $2M\geq \min\{2S+1,N\}$
	and thus $M\geq \min\{S+\frac{1}{2},\frac{1}{2}N\}$.
	Since $M$ needs to be an integer, this is equivalent to $M\geq \min\{S+1,\frac{1}{2}N\}$.
	\par
	$(2)\Rightarrow (1)$: Note that $2M\geq \min\{2S+2,N\}\geq\min\{2S+1,N\}$.
	Due to \thref{Theorem:signed_kernel_dimensions} there exists
	a matrix $\AMat\in\mathbb{R}^{2M\times N}$ with real signed kernel condition
	of order $S$. Define $\BMat\in\mathbb{K}^{M\times N}$ by letting its real part
	be the first $M$ rows of $\AMat$ and its imaginary part be the last $M$ rows of $\AMat$.
	It follows that $\Kernel{\BMat}\cap\mathbb{R}^N=\Kernel{\AMat}=\Kernel{\AMat}\cap\mathbb{R}^N$.
	Thus, $\BMat$ has complex signed kernel condition of order $S$.
	By \thref{Theorem:signed_kernel_condition} one has
	$\left\{\xVec\right\}=\argmin{\zVec\in\mathbb{R}_+^N}\normRHS{\BMat\zVec-\BMat\xVec}$
	for all $\xVec\in\Sigma_S^N\cap\mathbb{R}^N_+$
	and by \thref{Proposition:cone_closed} $\Cone{F-C}$ is closed.
	By \refP{Statement:Theorem:robustness_and_recovery:recovery=>robust} of
	\thref{Theorem:robustness_and_recovery} one gets
	$\kappa\left(\BMat\right)\geq\tau\left(\BMat\right)>0$.
	Setting $\BMat':=\frac{\BMat}{\kappa\left(\BMat\right)}$ yields
	$\tau_{\kappa}\left(\Sigma_S^N\cap\mathbb{R}_+^N,\mathbb{R}_+^N,\mathbb{C},\normLHS{\cdot},\normRHS{\cdot}\right)
		\geq\tau\left(\BMat'\right)=\frac{\tau\left(\BMat\right)}{\kappa\left(\BMat\right)}>0$.
\end{proof}
By using the matrix $\AMat$ from \thref{Theorem:real_signed_kernel:existence}
in the proof of \thref{Theorem:complex_signed_kernel_dimensions} one can show that matrices of the form
$\BMat_{m,n}=\aVec_m^{n-1}+i\aVec_{m+M}^{n-1}$ have a signed kernel condition if $M\geq S+1$
and all $\aVec_m>0$ are distinct. While these matrices have the optimal number of
measurements in this complex case, they are not Vandermonde matrices as in the next theorem.
\begin{Theorem}[Complex Vandermonde Matrices have $S=M-1$]\label{Theorem:complex_signed_kernel:existence}
	Let $\mathbb{K}=\mathbb{C}$, $M,N\in\mathbb{N}$ and set $N':=\max\{2M-N,0\}$.
	Let the following conditions hold true:
	\begin{enumerate}
		\item
			\label{Condition:Theorem:complex_signed_kernel:existence:positivity}
			Let $\aVec_m>0$ for all $m\in \SetOf{M}$.
		\item
			\label{Condition:Theorem:complex_signed_kernel:existence:arguments}
			Let the arguments suffice
			$\tVec_m\in\left(0,\frac{\pi}{N+N'+1-2M}\right)$ for all $m\in\SetOf{M}$.
		\item
			\label{Condition:Theorem:complex_signed_kernel:existence:distinctness}
			Let $\aVec_m\Exp{i\tVec_m}$ for all $m\in\SetOf{M}$ be distinct from each other.
	\end{enumerate}
	Let $\AMat\in\mathbb{C}^{M\times N}$
	have entries $\AMat_{m,n}=\aVec_m^{n-1+N'}\Exp{i\tVec_m\left(n-1+N'\right)}$.
	Then, $\AMat$ has the signed kernel condition of order $S:=M-1$.
	In particular, all conditions can be fulfilled
	together with the additional property
	$\abs{\AMat_{m,n}}=1$ by choosing
	and
	\begin{align}
		\nonumber
			\aVec_m
		=
			1
		\TextAnd
			\tVec_m
		=
			\frac{m}{M+1}\frac{\pi}{N+N'+1-2M}
	\end{align}
	for all $m\in\SetOf{M}$.
\end{Theorem}
\begin{proof}
	Let $\pVec\in\Kernel{\AMat}\cap\mathbb{R}^N\setminus\ZeroSet$ and consider the polynomial
	$\p{a}:=\sum_{n=1}^{N}\pVec_{n}a^{n-1+N'}$ which has real coefficients.
	Let $N''$ be the degree of $\p{}$.
	Since $0=\AMat\pVec=\left(\p{\aVec_m\Exp{i\tVec_m}}\right)_{m\in\SetOf{M}}$
	the polynomial $\p{}$ has the complex roots $\aVec_m\Exp{i\tVec_m}$ for $m\in\SetOf{M}$.
	Since $\p{}$ has real coefficients, $\p{}$ also has the complex roots
	$\aVec_m\Exp{i\left(-\tVec_m\right)}$ for $m\in\SetOf{M}$.
	Due to the assumed \refP{Condition:Theorem:complex_signed_kernel:existence:arguments}
	and \refP{Condition:Theorem:complex_signed_kernel:existence:distinctness}
	these $2M$ roots are all distinct from each other. Thus, these $2M$ roots suffice
	\refP{Condition:Theorem:complex_rule_of_signs:distinctness} and
	\refP{Condition:Theorem:complex_rule_of_signs:roots} of
	\thref{Theorem:complex_rule_of_signs}.
	Since $N''\leq N+N'-1$, one gets
	$\tVec_m\in\left(0,\frac{\pi}{N''+2-2M}\right)$.
	Hence all $2M$ roots suffice
	\refP{Condition:Theorem:complex_rule_of_signs:positivity} and
	\refP{Condition:Theorem:complex_rule_of_signs:arguments} of
	\thref{Theorem:complex_rule_of_signs}.
	If $N''<2M$, then $\p{}$ would have more roots than its degree
	and hence would need to be the zero polynomial. This would be a
	contradiction and one would get
	$\Kernel{\AMat}\cap\mathbb{R}^N=\left\{0\right\}$
	and that the signed kernel condition of any order would be fulfilled.
	Thus, assume that $N''\geq 2M$. It follows that $N''\geq 1$.
	If $2M\geq N$, then one gets the contradiction
	$2M+1\leq N''+1\leq N+N'=2M$. Hence, $2M+1\leq N$.
	By \thref{Theorem:complex_rule_of_signs} $\SetSize{\textnormal{var}\left(\pVec\right)}\geq 2M$.
	That means that $\pVec$ contains a subsequence of $2M+1$ elements
	so that subsequent elements are non-zero and have opposite sign.
	This subsequence and thus $\pVec$ need to have at least
	$M$ negative entries. Hence, $\AMat$ has signed kernel condition of
	order $M-1$.
\end{proof}
The advantage of this construction is that it is preserved under hermitian products.
\begin{Proposition}[Hermitian Products of Vandermonde Matrices have $S=\left\lceil\frac{1}{2}M^2\right\rceil-1$]
	\label{Proposition:complex_signed_kernel:hermitian_product}
	Let $\mathbb{K}=\mathbb{C}$, $M,N\in\mathbb{N}$, $\sigma:\SetOf{M}^2\rightarrow \SetOf{M^2}$
	be a bijection and set $N':=\max\{M^2-N,0\}$.
	Let the following conditions hold true:
	\begin{enumerate}
		\item
			\label{Condition:Proposition:complex_signed_kernel:hermitian_product:positivity}
			Let $\aVec_m>0$ for all $m\in \SetOf{M}$.
		\item
			\label{Condition:Proposition:complex_signed_kernel:hermitian_product:arguments}
			Let $\tVec_{m}\in\left(0,\frac{\pi}{N+N'+1-M^2}\right)$ for all $m\in\SetOf{M}$.
		\item
			\label{Condition:Proposition:complex_signed_kernel:hermitian_product:ordered}
			Let $\tVec_1<\dots<\tVec_M$.
		\item
			\label{Condition:Proposition:complex_signed_kernel:hermitian_product:distinctness}
			Let
			$\aVec_{m_1}\aVec_{m_2}\Exp{i\left(\tVec_{m_1}-\tVec_{m_2}\right)}$
			for all $m_1,m_2\in\SetOf{M}$ with $m_1\leq m_2$ be disctinct from each other.
	\end{enumerate}
	Let $\AMat\in\mathbb{C}^{M\times N}$
	have entries
	$\AMat_{m,n}=\aVec_{m}^{n-1+N'}\Exp{i\tVec_{m}\left(n-1+N'\right))}$
	and let $\AMat'\in\mathbb{C}^{M^2\times N}$ have entries
	$\AMat'_{\sigma\left(\left(m_k\right)_{k\in\SetOf{2}}\right),n}
		=\AMat_{m_1,n}\AMat_{m_2,n}^H$
	for all $m_1,m_2\in\SetOf{M},n\in\SetOf{N}$.
	In other words the $n$-th column of $\AMat'$ is a reordering of
	$\bVec\bVec^H$ where $\bVec$ is the $n$-th
	column of $\AMat$.
	Then, $\AMat'$ has the signed kernel condition of order
	$S:=\left\lceil\frac{1}{2}M^2\right\rceil-1$.
	\par
	In particular, the conditions can be fulfilled
	with $\abs{\AMat_{m,n}}\leq 1$ by choosing
	\begin{align}
		\nonumber
			\aVec_{m}
		:
			=m^{-\frac{1}{2}}
		\TextAnd
			\tVec_m
		:=
			\sqrt{\frac{\pi_m}{\pi_{M+1}}}\frac{\pi}{N+N'+1-M^2}
	\end{align}
	where $\pi_m$ is the $m$-th prime number.
\end{Proposition}
\begin{proof}
	Set $CR_{M,2}:=\SetSize{\left\{\left(m_k\right)_{k\in\SetOf{2}}:m_1,m_2\in\SetOf{M},m_2>m_1\right\}}$,
	$RR_{M,2}:=\#\big\{\left(m_k\right)_{k\in\SetOf{2}}:m_1,m_2\in\SetOf{M},m_2=m_1\big\}$
	and $R_{M,2}:=2CR_{M,2}+RR_{M,2}$.
	Note that $CR_{M,2}=\frac{1}{2}\left(M^2-M\right)$, $RR_{M,2}=M$
	and  $R_{M,2}=M^2$.
	\par
	Let $\pVec\in\Kernel{\AMat}\cap\mathbb{R}^N\setminus\ZeroSet$ and consider the polynomial
	$\p{a}:=\sum_{n=1}^{N}\pVec_{n}a^{n-1+N'}$ which has real coefficients.
	Let $N''$ be the degree of $\p{}$.
	Since $0=\left(\AMat\pVec\right)_{\sigma\left(\left(m_k\right)_{k\in\SetOf{2}}\right)}
		=\p{\aVec_{m_1}\aVec_{m_2}\Exp{i\left(\tVec_{m_1}-\tVec_{m_2}\right)}}$
	for all $m_1\geq m_2$, the polynomial $\p{}$ has the $CR_{M,2}$ complex roots
	$\aVec_{m_1}\aVec_{m_2}\Exp{i\left(\tVec_{m_1}-\tVec_{m_2}\right)}$ for all $m_1>m_2$ 
	and the $RR_{M,2}$ real roots
	$\aVec_{m_1}\aVec_{m_2}\Exp{i\left(\tVec_{m_1}-\tVec_{m_2}\right)}$ for $m_1=m_2$.
	Since $\p{}$ has real coefficients, $\p{}$ also has the $CR_{M,2}$ complex roots
	$\aVec_{m_1}\aVec_{m_2}\Exp{i\left(\tVec_{m_1}-\tVec_{m_2}\right)}$ for all $m_1<m_2$.
	Due to the assumed \refP{Condition:Proposition:complex_signed_kernel:hermitian_product:ordered} and
	\refP{Condition:Proposition:complex_signed_kernel:hermitian_product:arguments}
	one gets the angle condition
	\begin{align}
		\label{Equation:Proposition:complex_signed_kernel:hermitian_product:angles}
			0
		<
			\tVec_{m_1}-\tVec_{m_2}
		<
			\frac{\pi}{N+N'+1-R_{M,2}}
		<
			\pi
		\TextForAll
			m_1>m_2
	\end{align}
	so that roots $\Exp{i\left(\tVec_{m_1}-\tVec_{m_2}\right)}\aVec_{m_1}\aVec_{m_2}$
	with $m_1<m_2$ are distinct from these roots with $m_1>m_2$.
	Due to this and the assumed
	\refP{Condition:Proposition:complex_signed_kernel:hermitian_product:distinctness}
	these $R_{M,2}=2CR_{M,2}+RR_{M,2}$ roots given by
	$\Exp{i\left(\tVec_{m_1}-\tVec_{m_2}\right)}\aVec_{m_1}\aVec_{m_2}$ for $m_1,m_2\in\SetOf{M}$
	are all distinct from each other. Thus, these $R_{M,2}$ roots suffice
	\refP{Condition:Theorem:complex_rule_of_signs:distinctness} and
	\refP{Condition:Theorem:complex_rule_of_signs:roots} of
	\thref{Theorem:complex_rule_of_signs}.
	By $N''\leq N+N'-1$ and
	\eqref{Equation:Proposition:complex_signed_kernel:hermitian_product:angles}, one gets
	$\tVec_{m_1}-\tVec_{m_2}\in\left(-\frac{\pi}{N''+2-R_{M,2}},\frac{\pi}{N''+2-R_{M,2}}\right)$
	for all $m_1,m_2\in\SetOf{M}$.
	Hence all $R_{M,2}$ roots suffice
	\refP{Condition:Theorem:complex_rule_of_signs:positivity} and
	\refP{Condition:Theorem:complex_rule_of_signs:arguments} of
	\thref{Theorem:complex_rule_of_signs}.
	If $N''<R_{M,2}$, then $\p{}$ would have more roots than its degree
	and hence would need to be the zero polynomial. This would be a
	contradiction and one would get
	$\Kernel{\AMat}\cap\mathbb{R}^N=\left\{0\right\}$
	and that the signed kernel condition of any order would be fulfilled.
	Thus, assume that $N''\geq R_{M,2}$. It follows that $N''\geq 1$.
	If $R_{M,2}\geq N$, then one gets the contradiction
	$R_{M,2}+1\leq N''+1\leq N+N'=R_{M,2}$. Hence, $R_{M,2}+1\leq N$.
	By \thref{Theorem:complex_rule_of_signs} $\SetSize{\textnormal{var}\left(\pVec\right)}\geq R_{M,2}$.
	That means that $\pVec$ contains a subsequence of $R_{M,2}+1$ elements
	so that subsequent elements are non-zero and have opposite sign.
	This subsequence and thus $\pVec$ need to have at least
	$\left\lfloor\frac{1}{2}\left(R_{M,2}+1\right)\right\rfloor
		=\left\lceil\frac{1}{2}R_{M,2}\right\rceil$ negative entries.
	Hence, $\AMat$ has signed kernel condition of
	order $\left\lceil\frac{1}{2}R_{M,2}\right\rceil-1=\left\lceil\frac{1}{2}M^2\right\rceil-1$.
	\par
	For the in particular part it suffices to show that the choice suffices
	the conditions of this corollary. It is clear that
	\refP{Condition:Proposition:complex_signed_kernel:hermitian_product:positivity},
	\refP{Condition:Proposition:complex_signed_kernel:hermitian_product:arguments} and
	\refP{Condition:Proposition:complex_signed_kernel:hermitian_product:ordered}
	are fulfilled. For the distinctness condition let $m_1,m_2,m'_1,m'_2\in\SetOf{M}$
	with $m_1\leq m_2$, $m'_1\leq m'_2$ and
	\begin{align}
		\label{Equation:Proposition:complex_signed_kernel:hermitian_product:exp=exp}
			\aVec_{m_1}\aVec_{m_2}\Exp{i\left(\tVec_{m_1}-\tVec_{m_2}\right)}
		=
			\aVec_{m'_1}\aVec_{m'_2}\Exp{i\left(\tVec_{m'_1}-\tVec_{m'_2}\right)}.
	\end{align}
	By $\aVec_m>0$ this can only be fulfilled if there exists an $l\in\mathbb{Z}$ such that
	\begin{align}
		\label{Equation:Proposition:complex_signed_kernel:hermitian_product:zero+2pl}
			0
		=
			\tVec_{m_1}-\tVec_{m_2}-\tVec_{m'_1}+\tVec_{m'_2}+2\pi l
	\end{align}
	Due to
	\eqref{Equation:Proposition:complex_signed_kernel:hermitian_product:angles} one gets
	$\tVec_{m''_1}-\tVec_{m''_2}\in\left(0,\pi\right)$ for $\left(m''_1,m''_2\right)=\left(m_1,m_2\right)$
	and $\left(m''_1,m''_2\right)=\left(m'_1,m'_2\right)$. Thus, $l$ needs to be $0$.
	Plugging this and the choice of $\tVec_m$ into
	\eqref{Equation:Proposition:complex_signed_kernel:hermitian_product:zero+2pl}
	yields
	\begin{align}
		\nonumber
			0
		=&
			\sqrt{\frac{\pi_{m_1}}{\pi_{M+1}}}\frac{\pi}{N+N'+1-M^2}
			-\sqrt{\frac{\pi_{m_2}}{\pi_{M+1}}}\frac{\pi}{N+N'+1-M^2}
		\\&\nonumber
			-\sqrt{\frac{\pi_{m'_1}}{\pi_{M+1}}}\frac{\pi}{N+N'+1-M^2}
			+\sqrt{\frac{\pi_{m'_2}}{\pi_{M+1}}}\frac{\pi}{N+N'+1-M^2}.
	\end{align}
	Applying algebraic manipulation yields
	\begin{align}
		\label{Equation:Proposition:complex_signed_kernel:hermitian_product:linear_combination}
			0
		=
			\sqrt{\pi_{m_1}}-\sqrt{\pi_{m_2}}-\sqrt{\pi_{m'_1}}+\sqrt{\pi_{m'_2}}.
	\end{align}
	Note that
	$\left\{\sqrt{\pi_m}:m\in\mathbb{N}\right\}\subset\mathbb{R}$
	are linearly independent vectors in the vector space $\mathbb{R}$
	over $\mathbb{Q}$
	and by
	\eqref{Equation:Proposition:complex_signed_kernel:hermitian_product:linear_combination}
	zero is a linear combination of $\sqrt{\pi_{m_k}}$ and $\sqrt{\pi_{m_k'}}$
	for $k\in\SetOf{2}$ over $\mathbb{Q}$.
	Thus, all factors in front of all $\sqrt{\pi_m}$ must vanish.
	This can only be the case if either
	$m_1=m'_1$ and $m_2=m'_2$ or if
	$m_1=m_2$ and $m'_1=m'_2$.
	Suppose that $m_1=m_2$ and $m'_1=m'_2$. Plugging this and the choice of $\aVec_m$ into
	\eqref{Equation:Proposition:complex_signed_kernel:hermitian_product:exp=exp} yields
	\begin{align}
		\nonumber
			m_k^{-1}
		=
			\aVec_{m_1}\aVec_{m_2}
		=
			\aVec_{m'_1}\aVec_{m'_2}
		=
			{m'_k}^{-1}
	\end{align}
	for $k\in\SetOf{2}$. Thus in this case one also has $m_1=m'_1$ and $m_2=m'_2$.
	Hence, \refP{Condition:Proposition:complex_signed_kernel:hermitian_product:distinctness}
	is fulfilled.
\end{proof}
In applications normalized matrices are more common. These lose in the signed kernel condition since the diagonal elements can not correspond to distinct roots.
\begin{Proposition}[Hermitian Products of normalized Vandermonde Matrices have $S=\frac{1}{2}\left(M^2-M\right)$]
	\label{Proposition:complex_signed_kernel:hermitian_product_normalized}
	Let $\mathbb{K}=\mathbb{C}$, $M,N\in\mathbb{N}$, $\sigma:\SetOf{M}^2\rightarrow \SetOf{M^2}$
	be a bijection and set $N':=\max\{M^2-M+1-N,0\}$.
	Let the following conditions hold true:
	\begin{enumerate}
		\item
			\label{Condition:Proposition:complex_signed_kernel:hermitian_product_normalized:arguments}
			Let $\tVec_{m}\in\left(0,\frac{\pi}{N+N'-M^2+M}\right)$ for all $m\in\SetOf{M}$.
		\item
			\label{Condition:Proposition:complex_signed_kernel:hermitian_product_normalized:ordered}
			Let $\tVec_1<\dots<\tVec_M$.
		\item
			\label{Condition:Proposition:complex_signed_kernel:hermitian_product_normalized:distinctness}
			Let
			$\tVec_{m_1}-\tVec_{m_2}$
			for all $m_1,m_2\in\SetOf{M}$ with $m_1<m_2$ be disctinct from each other.
	\end{enumerate}
	Let $\AMat\in\mathbb{C}^{M\times N}$
	have entries
	$\AMat_{m,n}=\Exp{i\tVec_{m}\left(n-1+N'\right))}$
	and let $\AMat'\in\mathbb{C}^{M^2\times N}$ have entries
	$\AMat'_{\sigma\left(\left(m_k\right)_{k\in\SetOf{2}}\right),n}
		=\AMat_{m_1,n}\AMat_{m_2,n}^H$
	for all $m_1,m_2\in\SetOf{M},n\in\SetOf{N}$.
	In other words the $n$-th column of $\AMat'$ is a reordering of
	$\bVec\bVec^H$ where $\bVec$ is the $n$-th
	column of $\AMat$.
	Then, $\AMat'$ has the signed kernel condition of order
	$S:=\frac{1}{2}\left(M^2-M\right)$.
	\par
	In particular, the distinctness condition can be fulfilled
	by choosing
	\begin{align}
		\nonumber
			\tVec_m
		:=
			\sqrt{\frac{\pi_m}{\pi_{M+1}}}\frac{\pi}{N+N'-\left(M^2-M\right)+1}
	\end{align}
	where $\pi_m$ is the $m$-th prime number.
\end{Proposition}
\begin{proof}
	The proof proceeds similarly as the proof of \thref{Proposition:complex_signed_kernel:hermitian_product}
	however $RR_{M,2}$ is replaced by $1$ corresponding to the one real root generated by $m_1=1=m_2$
	for a total of $R_{M,2}=2CR_{M,2}+1=M^2-M+1$ disctint roots.
	The proof then follows from the fact that
	$\frac{1}{2}\left(R_{M,2}+1\right)$ is always an integer.
	The in particular part proceeds the same with the only difference being that
	the case $m_1=m_2$ and $m'_1=m'_2$ is not possible since the distinctness
	only needs to be fulfilled for $m_1<m_2$ and $m'_1<m'_2$.
\end{proof}
The signed kernel condition is even being preserved under $K$-fold outer products or simply tensors.
\begin{Corollary}[Column-wise Outer Products of Vandermonde Matrices have $S=\binom{M-1+K}{K}$]\label{Corollary:complex_signed_kernel:outer_product}
	\hspace{1mm}
	\\
	Let $\mathbb{K}=\mathbb{C}$, $M,N,K\in\mathbb{N}$, $\sigma:\SetOf{M}^K\rightarrow \SetOf{M^K}$
	be a bijection and set $N':=\max\{2\binom{M-1+K}{K}-N,0\}$.
	Let the following conditions hold true:
	\begin{enumerate}
		\item
			\label{Condition:Corollary:complex_signed_kernel:outer_product:positivity}
			Let $\aVec_m>0$ for all $m\in \SetOf{M}$.
		\item
			\label{Condition:Corollary:complex_signed_kernel:outer_product:arguments}
			Let the arguments suffice
			$\tVec_m\in\left(0,\frac{1}{K}\frac{\pi}{N+N'+1-2\binom{M-1+K}{K}}\right)$
			for all $m\in\SetOf{M}$.
		\item
			\label{Condition:Corollary:complex_signed_kernel:outer_product:distinctness}
			Let $\prod_{k=1}^K\aVec_{m_k}\Exp{i\sum_{k=1}^K\tVec_{m_k}}$ for all
			$m_k\in\SetOf{M},k\in\SetOf{K}$ with $m_1\leq\dots\leq m_K$ be distinct from each other.
	\end{enumerate}
	Let $\AMat\in\mathbb{K}^{M\times N}$
	have entries
	$\AMat_{m,n}=\aVec_{m}^{n-1+N'-N}\Exp{i\tVec_{m}\left(n-1+N'\right))}$
	and let $\AMat'\in\mathbb{K}^{M^K\times N}$ have entries
	$\AMat'_{\sigma\left(\left(m_k\right)_{k\in\SetOf{K}}\right),n}
		=\prod_{k=1}^K\AMat_{m_k,n}$
	for all $m_k\in\SetOf{M},k\in\SetOf{K},n\in\SetOf{N}$.
	In other words the $n$-th column of $\AMat'$ is the reordered
	outer product of the $n$-th column of $\AMat$ $K$-times with itself.
	Then, $\AMat'$ has the signed kernel condition of order
	$S:=\binom{M-1+K}{K}-1$.
	\par
	In particular, the conditions can be fulfilled
	with $\abs{\aVec_m}=1$ by choosing
	\begin{align}
		\nonumber
			\aVec_{m}
		:=
			1
		\TextAnd
			\tVec_m
		:=
			\sqrt{\frac{\pi_m}{\pi_{M+1}}}\frac{1}{K}\frac{\pi}{N+N'+1-2\binom{M-1+K}{K}}
	\end{align}
	where $\pi_m$ is the $m$-th prime number.
\end{Corollary}
\begin{proof}
	Let $CR_{M,K}:=\SetSize{\left\{\left(m_k\right)_{k\in\SetOf{K}}:m_k\in\SetOf{M},m_1\leq\dots\leq m_K\right\}}$.
	At first the combinatorial argument $CR_{M,K}=\binom{M-1+K}{K}$
	is shown. Note that $CR_{M,K}$ can be written as
	\begin{align}
		\label{Equation:Corollary:complex_signed_kernel:outer_product:CR}
			CR_{M,K}
		=
			\sum_{m_K=1}^M\sum_{m_{K-1}=1}^{m_K}\dots\sum_{m_1=1}^{m_2}1.
	\end{align}
	The rest follows from induction over $K$.
	For an induction start note that \eqref{Equation:Corollary:complex_signed_kernel:outer_product:CR} yields
	\begin{align}
		\nonumber
			CR_{M,1}
		=
			\sum_{m=1}^M1
		=
			M
		=
			\binom{M-1+1}{1}.
	\end{align}
	Thus, the claim holds for $K=1$ and all $M\in\mathbb{N}$.
	Now assume that $CR_{M,K}=\binom{M-1+K}{K}$ holds for some $K\in\mathbb{N}$ and all $M\in\mathbb{N}$.
	By \eqref{Equation:Corollary:complex_signed_kernel:outer_product:CR}
	and the induction hypothesis one gets
	\begin{align}
		\nonumber
			CR_{M,K+1}
		=
			\sum_{m=1}^M\sum_{m_{K}=1}^{m}\sum_{m_{K-1}=1}^{m_K}\dots\sum_{m_1=1}^{m_2}1
		=
			\sum_{m=1}^MCR_{m,K}
		=
			\sum_{m=1}^M\binom{m-1+K}{K}
		=
			\sum_{m=K}^{M-1+K}\binom{m}{K}.
	\end{align}
	Applying \cite[Theorem~1.5.2]{chus_theorem} yields 
	\begin{align}
		\nonumber
			CR_{M,K+1}
		=
			\sum_{m=K}^{M-1+K}\binom{m}{K}
		=
			\binom{M-1+K+1}{K+1}
	\end{align}
	which completes the induction. Hence,
	\begin{align}
		\nonumber
		CR_{M,K}=\binom{M-1+K}{K} \TextForAll K,M\in\mathbb{N}.
	\end{align}
	\par
	Set $\tVec'_{\sigma\left(\left(m_k\right)_{k\in\SetOf{K}}\right)}:=\sum_{k=1}^K\tVec_{m_k}$
	and $\aVec'_{\sigma\left(\left(m_k\right)_{k\in\SetOf{K}}\right)}:=\prod_{k=1}^K\aVec_{m_k}$
	so that
	\begin{align}
		\nonumber
			\AMat'_{m,n}={\aVec'}_{m}^{n-1+N'}\Exp{i\tVec'_{m}\left(n-1+N'\right)}
		\TextForAll
			m\in\SetOf{M^K},n\in\SetOf{N}.
	\end{align}
	At first glance $\AMat'$
	is of the same structure as in
	\thref{Theorem:complex_signed_kernel_dimensions}
	however this is not true. Since permutations of
	$\sigma\left(\left(m_k\right)_{k\in\SetOf{K}}\right)$
	create the same products
	$\prod_{k=1}^K\aVec_{m_k}\Exp{i\sum_{k=1}^K\tVec_{m_k}}$,
	the distinctness condition is not given for all $m_k\in\SetOf{M}$, $k\in\SetOf{K}$.
	However, the assumed \refP{Condition:Corollary:complex_signed_kernel:outer_product:distinctness}
	yields that $\prod_{k=1}^K\aVec_{m_k}\Exp{i\sum_{k=1}^K\tVec_{m_k}}$ for $m_k\in\SetOf{M}$
	with $m_1\leq\dots\leq m_K$ are distinct from each other.
	Thus, let $\AMat''$ be the matrix consisting of the $CR_{M,K}$ rows
	of $\AMat'$ with indices
	$\sigma\left(\left(m_k\right)_{k\in\SetOf{K}}\right)$
	for all $m_k\in\SetOf{M}$ with $m_1\leq\dots\leq m_K$.
	The matrix $\AMat''\in\mathbb{K}^{CR_{M,K}\times N}$
	suffices \refP{Condition:Theorem:complex_signed_kernel:existence:distinctness} of
	\thref{Theorem:complex_signed_kernel:existence}.
	By the assumed \refP{Condition:Corollary:complex_signed_kernel:outer_product:positivity}
	and \refP{Condition:Corollary:complex_signed_kernel:outer_product:arguments}
	it follows that
	$\aVec'_m>0$ and
	$\tVec'_m\in\left(0,\frac{\pi}{N+N'+1-2CR_{M,K}}\right)$
	for all $m\in\SetOf{M^K}$.
	Hence \refP{Condition:Theorem:complex_signed_kernel:existence:positivity} and
	\refP{Condition:Theorem:complex_signed_kernel:existence:arguments} of
	\thref{Theorem:complex_signed_kernel:existence} are fulfilled.
	Lastly it holds true that $N'=\max\{2CR_{M,K}-N,0\}$.
	Thus, all conditions of \thref{Theorem:complex_signed_kernel:existence}
	are fulfilled and $\AMat''$ has signed kernel condition of order
	$CR_{M,K}-1$. Since $\AMat'$
	is $\AMat''$ just with some added rows that are all
	linearly dependent to rows of $\AMat''$, the
	kernels of $\AMat'$ and $\AMat''$ are the same.
	Thus, also $\AMat'$ has	the signed kernel condition of order
	$CR_{M,K}-1=\binom{M-1+K}{K}-1$.
	\par
	For the in particular part it suffices to show that the choice suffices
	the conditions of this corollary. It is clear that
	\refP{Condition:Corollary:complex_signed_kernel:outer_product:positivity} and
	\refP{Condition:Corollary:complex_signed_kernel:outer_product:arguments}
	are fulfilled. For the distinctness condition let
	$m_k,m_k'\in\SetOf{M}$, $k\in\SetOf{K}$ with
	$m_1\leq\dots\leq m_K$ and $m_1'\leq\dots\leq m_K'$ be such that
	\begin{align}
		\nonumber
			\prod_{k=1}^K\aVec_{m_k}\Exp{i\sum_{k=1}^K\tVec_{m_k}}
		=
			\prod_{k=1}^K\aVec_{m_k'}\Exp{i\sum_{k=1}^K\tVec_{m_k'}}.
	\end{align}
	Plugging in the choice $\aVec_m=1$ and applying algebraic maniuplation
	yields that there exists a $l\in\mathbb{Z}$ such that
	\begin{align}
		\label{Equation:Corollary:complex_signed_kernel:outer_product:zero+2pl}
			0
		=
			\sum_{k=1}^K\tVec_{m_k}
			+\sum_{k=1}^K\left(-\tVec_{m_k'}\right)
			+2\pi l
	\end{align}
	Since $\sqrt{\frac{\pi_m}{\pi_{M+1}}}\in\left(0,1\right)$,
	one gets $\sum_{k=1}^K\tVec_{m_k''}\in\left(0,\frac{\pi}{N+N'+1-2\binom{M-1+K}{K}}\right)\subset\left(0,\pi\right)$ for $m_k''=m_k$, $k\in\SetOf{K}$ and for $m_k''=m_k'$, $k\in\SetOf{K}$.
	Thus, $l$ needs to be $0$.
	Plugging this and the choice of $\tVec_m$ into
	\eqref{Equation:Corollary:complex_signed_kernel:outer_product:zero+2pl}
	yields
	\begin{align}
		\nonumber
			0
		=
			\sum_{k=1}^K
			\sqrt{\frac{\pi_{m_k}}{\pi_{M+1}}}\frac{1}{K}\frac{\pi}{N+N'+1-2\binom{M-1+K}{K}}
			+\sum_{k=1}^K
			\left(-\sqrt{\frac{\pi_{m_k'}}{\pi_{M+1}}}\frac{1}{K}\frac{\pi}{N+N'+1-2\binom{M-1+K}{K}}\right).
	\end{align}
	Applying more algebraic manipulation yields
	\begin{align}
		\nonumber
			0
		=
			\sum_{k=1}^K\sqrt{\pi_{m_k}}
			+\sum_{k=1}^K\left(-1\right)\sqrt{\pi_{m_k'}}.
	\end{align}
	Rewriting the sums gives
	\begin{align}
		\label{Equation:Corollary:complex_signed_kernel:outer_product:linear_combination}
			0
		=
			\sum_{m=1}^M\left(\sum_{k\in\SetOf{K}:m_k=m}1-\sum_{k\in\SetOf{K}:m'_k=m}1\right)\sqrt{\pi_m}
	\end{align}
	Note that
	$\left\{\sqrt{\pi_m}:m\in\mathbb{N}\right\}\subset\mathbb{R}$
	are linearly independent vectors in the vector space $\mathbb{R}$
	over $\mathbb{Q}$
	and by
	\eqref{Equation:Corollary:complex_signed_kernel:outer_product:linear_combination}
	zero is a linear combination of $\sqrt{\pi_{m_k}}$ and $\sqrt{\pi_{m_k'}}$
	for $k\in\SetOf{K}$ over $\mathbb{Q}$.
	Thus, $\sum_{k\in\SetOf{K}:m_k=m}1-\sum_{k\in\SetOf{K}:m'_k=m}1=0$
	and hence $\SetSize{\left\{k\in\SetOf{K}:m_k=m\right\}}=\SetSize{\left\{k\in\SetOf{K}:m'_k=m\right\}}$
	for all $m\in\SetOf{M}$.
	Since $m_1\leq\dots\leq m_K$ and $m_1'\leq\dots\leq m_K'$,
	this can only happen if $m_k=m'_k$ for all $k\in\SetOf{K}$.
	Thus, \refP{Condition:Corollary:complex_signed_kernel:outer_product:distinctness}
	is fulfilled.
\end{proof}
Note that the resulting signed kernel condition
in \thref{Corollary:complex_signed_kernel:outer_product} with $K=2$
is higher than the one obtained from \thref{Proposition:complex_signed_kernel:hermitian_product}.
Due to the complex conjugate in \thref{Proposition:complex_signed_kernel:hermitian_product}
the diagonal elements are forced to be real number which in turn correspond to real roots
which can not be counted twice.
\begin{Remark}
	Let $\mathbb{K}=\mathbb{C}$, $\AMat\in\mathbb{C}^{M\times N}$ and
	$\sigma:\SetOf{M}^2\rightarrow\SetOf{M^2}$ be a bijection.
	Let $\AMat'\in\mathbb{C}^{M^2\times N}$ such that
	$\AMat'_{\sigma\left(m_1,m_2\right),n}=\AMat_{m_1,n}\AMat_{m_2,n}^H$ for all
	$m_1,m_2\in\SetOf{M}$, $n\in\SetOf{N}$ such that $\AMat'$ has the signed
	kernel condition of order $S\in\mathbb{N}$.
	Then, $\min\left\{2S+1,N\right\}\leq M^2$.
	In particular, if $2S+1\leq N$, then $S\leq\left\lceil\frac{1}{2}M^2\right\rceil-1$.
	Thus, the construction in \thref{Proposition:complex_signed_kernel:hermitian_product}
	has the maximal possible order of a signed kernel condition.
\end{Remark}
\begin{proof}
	For this proof it is required to
	distinguish between the real signed kernel condition and the complex
	signed kernel condition again.
	Consider the matrix $\BMat\in\mathbb{R}^{2M^2\times N}$ where the first $M^2$ rows
	be the real part of $\AMat'$ and the last $M^2$ rows of $\BMat$ be the imaginary part of
	$\AMat'$. It follows that
	$\Kernel{\BMat}=\Kernel{\BMat}\cap\mathbb{R}^N=\Kernel{\AMat}\cap\mathbb{R}^N$.
	Now set $CR_{M,2}:=\SetSize{\left\{\left(m_k\right)_{k\in\SetOf{2}}:m_1,m_2\in\SetOf{M},m_2>m_1\right\}}$,
	$RR_{M,2}:=\SetSize{\left\{\left(m_k\right)_{k\in\SetOf{2}}:m_1,m_2\in\SetOf{M},m_2=m_1\right\}}$.
	Since $\AMat'$ is hermitian, it follows that $\BMat$ has maximally
	$2CR_{M,2}+RR_{M,2}$
	linearly independent rows, two for each element in $CR_{M,2}$ and one for each element
	in $RR_{M,2}$. Let $\BMat'\in\mathbb{R}^{2CR_{M,2}+RR_{M,2},N}$
	be a matrix of those rows. Since only linearly dependent rows were removed
	it follows that
	$\Kernel{\BMat'}\cap\mathbb{R}^N=\Kernel{\BMat}\cap\mathbb{R}^N
		=\Kernel{\AMat'}\cap\mathbb{R}^N$.
	Thus, $\BMat'$ has the real signed kernel condition of order $S$. 
	By \thref{Theorem:signed_kernel_dimensions} it follows that
	$M^2=2CR_{M,2}+RR_{M,2}\geq \min\left\{2S+1,N\right\}$.
	Now if $2S+1\leq N$, then $S\leq \frac{1}{2}\left(M^2-1\right)$.
	Since $S$ is an integer, it follows that
	$S\leq\left\lfloor \frac{1}{2}\left(M^2-1\right)\right\rfloor
		=\left\lceil\frac{1}{2}M^2\right\rceil-1$.
\end{proof}
Note that similar optimality results can be drawn for the outer product construction
and the other hermitian construction in this subsection.
	\section{Outwardly Neighborly Polytopes}
In this section the case $\mathbb{K}=\mathbb{R}$,
$C=\eta\left(\Sigma_S^N\cap\mathbb{S}^{N-1}_{\norm{\cdot}_1}\cap\mathbb{R}_+^N\right)$
and $F=\eta\left(\mathbb{B}^{N}_{\norm{\cdot}_1}\cap\mathbb{R}_+^N\right)$,
where $\eta>0$, and its relation to the signed kernel condition case is discussed.
In \cite{outwardly_neighborly_polytopes,noneg_nsp} an optimization problem is considered
which is basically the case with these $F$ and $C$.
\begin{Theorem}\label{Theorem:outwardly_neighborly_polytopes}
	Let $\mathbb{K}=\mathbb{R}$, $\normRHS{\cdot}$ be any norm on $\mathbb{R}^M$.
	Let $S\in\mathbb{N}$ and $\AMat\in\mathbb{R}^{M\times N}$ and $\eta>0$.
	Consider the following conditions:
	\begin{enumerate}
		\item
			\label{Condition:Theorem:outwardly_neighborly_polytopes:FC}
			For all
			$\xVec\in\eta\left(\Sigma_S^N\cap\mathbb{S}^{N-1}_{\norm{\cdot}_1}\cap\mathbb{R}_+^N\right)$ one has
			$
					\left\{\xVec\right\}
				=
					\argmin{\zVec\in\eta\left(\mathbb{B}^{N}_{\norm{\cdot}_1}\cap\mathbb{R}_+^N\right)}
					\normRHS{\AMat\zVec-\AMat\xVec}
			$.
		\item
			\label{Condition:Theorem:outwardly_neighborly_polytopes:LP}
			For all
			$\xVec\in\Sigma_S^N\cap\mathbb{R}_+^N$ one has
			$
					\left\{\xVec\right\}
				=
					\argmin{\zVec\in\mathbb{R}_+^N:\AMat\zVec=\AMat\xVec}
					\norm{\zVec}_1
			$.
		\item
			\label{Condition:Theorem:outwardly_neighborly_polytopes:NNLR}
			For all
			$\xVec\in\Sigma_S^N\cap\mathbb{R}_+^N$ one has
			$
					\left\{\xVec\right\}
				=
					\argmin{\zVec\in\mathbb{R}_+^N}
					\normRHS{\AMat\zVec-\AMat\xVec}
			$.
	\end{enumerate}
	Then, the following statements are true:
	\begin{enumerate}
		\item
			\label{Statement:Theorem:outwardly_neighborly_polytopes:equivalence}
			\refP{Condition:Theorem:outwardly_neighborly_polytopes:FC},
			\refP{Condition:Theorem:outwardly_neighborly_polytopes:LP}
			are equivalent.
		\item
			\label{Statement:Theorem:outwardly_neighborly_polytopes:NNLR=>others}
			\refP{Condition:Theorem:outwardly_neighborly_polytopes:NNLR}
			implies
			\refP{Condition:Theorem:outwardly_neighborly_polytopes:FC} and
			\refP{Condition:Theorem:outwardly_neighborly_polytopes:LP}.
		\item
			\label{Statement:Theorem:outwardly_neighborly_polytopes:others=>NNLR}
			Let additionally there be $\vVec\in\mathbb{R}^M$ such that
			$\AMat^H\vVec$ is the all ones vector in $\mathbb{R}^N$.
			Then,
			\refP{Condition:Theorem:outwardly_neighborly_polytopes:FC} and
			\refP{Condition:Theorem:outwardly_neighborly_polytopes:LP}
			imply
			\refP{Condition:Theorem:outwardly_neighborly_polytopes:NNLR}.
	\end{enumerate}
\end{Theorem}
\begin{proof}
	\refP{Statement:Theorem:outwardly_neighborly_polytopes:equivalence}:
	\refP{Condition:Theorem:outwardly_neighborly_polytopes:LP}
	$\Rightarrow$
	\refP{Condition:Theorem:outwardly_neighborly_polytopes:FC}:
	Let $\xVec\in\eta\left(\Sigma_S^N\cap\mathbb{S}^{N-1}_{\norm{\cdot}_1}\cap\mathbb{R}_+^N\right)$ and $\xVec^\#$ be a minimizer of
	$\min_{\zVec\in\eta\left(\mathbb{B}^{N}_{\norm{\cdot}_1}\cap\mathbb{R}_+^N\right)}\normRHS{\AMat\zVec-\AMat\xVec}$.
	Since $\xVec$ is feasible for this problem,
	$\normRHS{\AMat\xVec^\#-\AMat\xVec}=0$
	and thus $\xVec^\#$ is feasible for
	$\min_{\zVec\in\mathbb{R}_+^N:\AMat\zVec=\AMat\xVec}\norm{\zVec}_1$.
	Since $\norm{\xVec^\#}_1\leq\eta=\norm{\xVec}_1$
	and $\xVec$ is a minimizer by assumption,
	$\xVec^\#$ is also minimizer for 
	$\min_{\zVec\in\mathbb{R}_+^N:\AMat\zVec=\AMat\xVec}\norm{\zVec}_1$.
	By assumption $\xVec^\#=\xVec$.
	\par
	\refP{Condition:Theorem:outwardly_neighborly_polytopes:FC}
	$\Rightarrow$
	\refP{Condition:Theorem:outwardly_neighborly_polytopes:LP}:
	Let $\xVec\in\Sigma_S^N\cap\mathbb{R}_+^N$ and
	$\xVec^\#\in\argmin{\zVec\in\mathbb{R}_+^N:\AMat\zVec=\AMat\xVec}\norm{\zVec}_1$.
	If $\xVec=0$, the feasibility of $\xVec$ yields $\norm{\xVec^\#}_1\leq\norm{\xVec}_1=0$
	and thus $\xVec^\#=0$. Now assume $\xVec\neq 0$.
	Since $\xVec$ is feasible for $\min_{\zVec\in\mathbb{R}_+^N:\AMat\zVec=\AMat\xVec}\norm{\zVec}_1$, one gets
	$\norm{\frac{\eta}{\norm{\xVec}_1}\xVec^\#}_1\leq\frac{\eta}{\norm{\xVec}_1}\norm{\xVec}_1=\eta$.
	Thus, $\frac{\eta}{\norm{\xVec}_1}\xVec^\#$ is also feasible for
	$\min_{\zVec\in\eta\left(\mathbb{B}^{N}_{\norm{\cdot}_1}\cap\mathbb{R}_+^N\right)}\normRHS{\AMat\zVec-\AMat\frac{\eta}{\norm{\xVec}_1}\xVec}$.
	Since $\xVec^\#$ is feasible for
	$\min_{\zVec\in\mathbb{R}_+^N:\AMat\zVec=\AMat\xVec}\norm{\zVec}_1$,
	one gets
	$\normRHS{\AMat\frac{\eta}{\norm{\xVec}_1}\xVec^\#-\AMat\frac{\eta}{\norm{\xVec}_1}\xVec}=0$.
	Thus, $\frac{\eta}{\norm{\xVec}_1}\xVec^\#$ is a minimizer of the problem
	$\min_{\zVec\in\eta\left(\mathbb{B}^{N}_{\norm{\cdot}_1}\cap\mathbb{R}_+^N\right)}\normRHS{\AMat\zVec-\AMat\frac{\eta}{\norm{\xVec}_1}\xVec}$
	and by assumption $\frac{\eta}{\norm{\xVec}_1}\xVec=\frac{\eta}{\norm{\xVec}_1}\xVec^\#$.
	\par
	\refP{Statement:Theorem:outwardly_neighborly_polytopes:NNLR=>others}:
	\refP{Condition:Theorem:outwardly_neighborly_polytopes:NNLR}
	$\Rightarrow$
	\refP{Condition:Theorem:outwardly_neighborly_polytopes:LP}:
	Let $\xVec\in\Sigma_S^N\cap\mathbb{R}_+^N$ and
	$\xVec^\#\in\argmin{\zVec\in\mathbb{R}_+^N:\AMat\zVec=\AMat\xVec}\norm{\zVec}_1$.
	Since $\xVec^\#$ is feasible for this problem, it is also feasible for
	$\min_{\zVec\in\mathbb{R}_+^N}\normRHS{\AMat\zVec-\AMat\xVec}$
	and actually a minimizer of that. By assumption $\xVec^\#=\xVec$.
	\par
	\refP{Statement:Theorem:outwardly_neighborly_polytopes:others=>NNLR}:
	Let additionally there be $\vVec\in\mathbb{R}^M$ such that
	$\wVec:=\AMat^H\vVec$ is the all ones vector in $\mathbb{R}^N$.
	Now 
	\refP{Condition:Theorem:outwardly_neighborly_polytopes:LP}
	$\Rightarrow$
	\refP{Condition:Theorem:outwardly_neighborly_polytopes:NNLR}
	is being shown.
	Let $\xVec\in\Sigma_S^N\cap\mathbb{R}_+^N$ and
	$\xVec^\#\in\argmin{\zVec\in\mathbb{R}_+^N}\normRHS{\AMat\zVec-\AMat\xVec}$.
	Since $\xVec$ is feasible for this, one gets
	$\AMat\xVec^\#=\AMat\xVec$. On the one hand this yields that
	$\xVec^\#$ is feasible for
	$\min_{\zVec\in\mathbb{R}_+^N:\AMat\zVec=\AMat\xVec}\norm{\zVec}_1$.
	And on the other hand this and the non-negativity of $\xVec$ and
	$\xVec^\#$ gives
	\begin{align}
		\nonumber
			\norm{\xVec^\#}_1-\norm{\xVec}_1
		=
			\scprod{\wVec}{\xVec^\#}-\scprod{\wVec}{\xVec}
		=
			\scprod{\wVec}{\xVec^\#-\xVec}
		=
			\scprod{\AMat^H\vVec}{\xVec^\#-\xVec}
		=
			\scprod{\vVec}{\AMat\xVec^\#-\AMat\xVec}
		=
			0.
	\end{align}
	Since by assumption $\xVec$ is a minimizer, $\xVec^\#$ must also be a minimizer of 
	$\min_{\zVec\in\mathbb{R}_+^N:\AMat\zVec=\AMat\xVec}\norm{\zVec}_1$
	and by assumption
	$\xVec^\#=\xVec$.
\end{proof}
In \cite{outwardly_neighborly_polytopes} it was proven that the recovery in this case
is equivalent to a property related to a polytope being outwardly $S$-neighborly.
Note that the condition on the all ones vector is mandatory in
\refP{Statement:Theorem:outwardly_neighborly_polytopes:others=>NNLR}.
\begin{Example}
	Let $\AMat=\begin{pmatrix}1&0&1\\0&1&1\end{pmatrix}$,
	then there exists no $\vVec\in\mathbb{R}^2$ such that
	$\AMat^H\vVec$ is the all ones vector!
	Further,
	\refP{Condition:Theorem:outwardly_neighborly_polytopes:FC} and
	\refP{Condition:Theorem:outwardly_neighborly_polytopes:LP}
	are fulfilled because
	$\Kernel{\AMat}=\left\{\alpha\begin{pmatrix}1\\1\\-1\end{pmatrix}:\alpha\in\mathbb{R}\right\}$,
	but
	\refP{Condition:Theorem:outwardly_neighborly_polytopes:NNLR}
	is not fulfilled due to \thref{Theorem:signed_kernel_dimensions}.
\end{Example}
Due to this example, the outwardly neighborly polytope case and the
signed kernel condition case are indeed different. However, the cases are related. In
\cite[Corollary~1.1]{outwardly_neighborly_polytopes} it was proven that recovery in the
outwardly neighborly polytope case is possible when $M\geq 2S$
only requiring one measurement less as in signed kernel condition case.
This can be used to create matrices with a signed kernel condition.
\begin{Remark}
	\label{Remark:row_of_ones}
	If $\AMat\in\mathbb{R}^{M\times N}$ suffices
	\refP{Condition:Theorem:outwardly_neighborly_polytopes:FC}
	or \refP{Condition:Theorem:outwardly_neighborly_polytopes:LP}
	of \thref{Theorem:outwardly_neighborly_polytopes}.
	then $\AMat$ with an additional row of ones suffices
	\refP{Condition:Theorem:outwardly_neighborly_polytopes:NNLR}
	of \thref{Theorem:outwardly_neighborly_polytopes}.
\end{Remark}
\begin{proof}
	Due to \thref{Theorem:robustness_and_recovery}
	\refP{Condition:Theorem:outwardly_neighborly_polytopes:FC} of
	\thref{Theorem:outwardly_neighborly_polytopes} is equivalent to
	$\left(F-C\right)\cap\Kernel{\AMat}=\ZeroSet$ for a general $\AMat$. This property stays valid
	if a row of ones gets added to $\AMat$. Due to
	\refP{Statement:Theorem:outwardly_neighborly_polytopes:others=>NNLR}
	of \thref{Theorem:outwardly_neighborly_polytopes}
	the extended matix suffices
	\refP{Condition:Theorem:outwardly_neighborly_polytopes:NNLR}
	of \thref{Theorem:outwardly_neighborly_polytopes}.
\end{proof}
Note that the statement about the row of ones from
\thref{Remark:row_of_ones} and \refP{Statement:Theorem:outwardly_neighborly_polytopes:others=>NNLR}
from \thref{Theorem:outwardly_neighborly_polytopes}
are quite similar to a result in
\cite[Theorem~3.4]{NNLAD}. In \cite[Theorem~3.4]{NNLAD} an additional
row of ones would yield a perfect constant for an $M^+$ criterion and thus
transfers a robust recovery guarantee from null space property case to
a robust recovery guarantee in the signed kernel condition case.
	\section{Bound for the Complexity Constants}
In this section a bound for the complexity constant is derived.
It is being derived using a similar methods as the proof of
\cite[Theorem~11.7]{IntroductionCS}.
At first a well known counting argument
based on volumes is needed.
\begin{Lemma}\label{Lemma:distance_in_unit_ball}
	Let $\normRHS{\cdot}$ be a norm on $\mathbb{K}^M$ and
	$W\subset \mathbb{B}_{\normRHS{\cdot}}^M$ be a finite set
	with at least two elements. Let $c'=1$ if $\mathbb{K}=\mathbb{R}$ and $c'=2$ if $\mathbb{K}=\mathbb{C}$.
	Then,
	\begin{align}
		\nonumber
			\inf_{\wVec,\wVec'\in W:\wVec\neq \wVec'}\normRHS{\wVec-\wVec'}
		\leq&
			2\left(\Exp{\frac{1}{c'M}\Ln{\SetSize{W}}}-1\right)^{-1}.
	\end{align}
\end{Lemma}
\begin{proof}
	Let $\mathbb{K}=\mathbb{R}$ and
	$0<t<\frac{1}{2}\inf_{\wVec,\wVec'\in W:\wVec\neq \wVec'}\normRHS{\wVec-\wVec'}$
	be arbitrary.
	Let $\textnormal{vol}\left[\cdot\right]$ denote the volume of an
	integrable set in the Lebesque measure.
	Consider the balls $\mathbb{B}_{\normRHS{\cdot}}^M\left(\wVec,t\right)
	:=\wVec+t\mathbb{B}_{\normRHS{\cdot}}^M$
	for all $\wVec\in W$
	which are disjoint by definition of $t$.
	Due to $W\subset \mathbb{B}_{\normRHS{\cdot}}^M$
	one has
	$\bigcup_{\wVec\in W}\mathbb{B}_{\normRHS{\cdot}}^M\left(\wVec,t\right)
		\subset \mathbb{B}_{\normRHS{\cdot}}^M\left(0,1+t\right)$.
	Since the union is taken over disjoint sets, it follows that
	\begin{align}
		\nonumber
			\SetSize{W} t^M
			\textnormal{vol}\left[\mathbb{B}_{\normRHS{\cdot}}^M\left(0,1\right)\right]
		=&
			\sum_{\wVec\in W}t^M
			\textnormal{vol}\left[\mathbb{B}_{\normRHS{\cdot}}^M\left(0,1\right)\right]
		=
			\sum_{\wVec\in W}
			\textnormal{vol}\left[\mathbb{B}_{\normRHS{\cdot}}^M\left(\wVec,t\right)\right]
		\\=&\nonumber
			\textnormal{vol}\left[\bigcup_{\wVec\in W}\mathbb{B}_{\normRHS{\cdot}}^M\left(\wVec,t\right)\right]
		\leq
			\textnormal{vol}\left[\mathbb{B}_{\normRHS{\cdot}}^M\left(0,1+t\right)\right]
		=
			\left(1+t\right)^M
			\textnormal{vol}\left[\mathbb{B}_{\normRHS{\cdot}}^M\left(0,1\right)\right].
	\end{align}
	By the fact that
	$\textnormal{vol}\left[\mathbb{B}_{\normRHS{\cdot}}^M\left(0,1\right)\right]$
	is non-zero, one gets $\SetSize{W} t^M\leq\left(1+t\right)^M$
	and thus $t\leq \left(\SetSize{W}^\frac{1}{M}-1\right)^{-1}$.
	Since this holds for all $0<t<\frac{1}{2}\inf_{\wVec,\wVec'\in W:\wVec\neq \wVec'}\normRHS{\wVec-\wVec'}$,
	it follows that
	\begin{align}
			\frac{1}{2}\inf_{\wVec,\wVec'\in W:\wVec\neq \wVec'}
			\normRHS{\wVec-\wVec'}
		\leq&
			\left(\SetSize{W}^\frac{1}{M}-1\right)^{-1}
		=
			\left(\Exp{\frac{1}{M}\Ln{\SetSize{W}}}-1\right)^{-1}.
	\end{align}
	If $\mathbb{K}=\mathbb{C}$, the proof follows from the fact that
	the complex unit ball $\mathbb{B}_{\normRHS{\cdot}}^M$ is a real unit ball in
	$\mathbb{R}^{2M}$ for a certain norm.
\end{proof}
The following theorem gives a bound for the complexity constant
in several important cases.
\begin{Theorem}\label{Theorem:bound_complexity_constant}
	Let $M,N,S\in\mathbb{N}$, $q\in\left[1,\infty\right]$
	and set $\frac{1}{q}:=0$ if $q=\infty$.
	Let $\normLHS{\cdot}=\norm{\cdot}_q$ be an $\ell_q$-norm
	and $\normRHS{\cdot}$ a norm on $\mathbb{K}^M$.
	Let one of the following hold true:
	\begin{enumerate}
		\item
			$C=\Sigma_S^N$ and $F=\Sigma_S^N$ (the Kruskal Rank case).
		\item
			$C=\Sigma_S^N\cap\mathbb{R}^N_+$ and $F=\mathbb{R}^N_+$
			(the Signed Kernel Condition case).
		\item
			$C=\eta\left(\Sigma_S^N\cap\mathbb{S}^{N-1}_{\norm{\cdot}_1}\right)$
			and $F=\eta\mathbb{B}^{N}_{\norm{\cdot}_1}$,
			where $\eta>0$.
			(the Null Space Property case).
		\item
			$C=\eta\left(\Sigma_S^N\cap\mathbb{S}^{N-1}_{\norm{\cdot}_1}\cap\mathbb{R}_+^N\right)$
			and $F=\eta\left(\mathbb{B}^{N}_{\norm{\cdot}_1}\cap\mathbb{R}_+^N\right)$,
			where $\eta>0$ (Outwardly Neighborly Polytope case).
	\end{enumerate}
	Let $c=2$ if $\mathbb{K}=\mathbb{R}$ and $c=4$ if $\mathbb{K}=\mathbb{C}$.
	Then,
	\begin{align}
		\nonumber
			\tau_{\kappa}\left(C,F,\mathbb{K},\norm{\cdot}_q,\normRHS{\cdot}\right)
		\leq
			2\left(\frac{2}{3}\right)^\frac{1}{q}
			\left(\Exp{\frac{S}{cM}\Ln{\frac{N}{4S}}}-1\right)^{-1}.
	\end{align}
\end{Theorem}
\begin{proof}
	Let $\AMat\in\mathbb{K}^{M\times N}$ with $\kappa\left(\AMat\right)\leq 1$.
	Let $T_1,\dots,T_L$ be the sets from \cite[Lemma~10.12]{IntroductionCS}
	with
	\begin{align}
		\label{Equation:Theorem:bound_complexity_constant:definition:sets}
			\SetSize{T_l}=S
		\TextAnd
			\SetSize{T_l\cap T_{l'}}<\frac{S}{2}
		\TextAnd
			L\geq \left(\frac{N}{4S}\right)^{\frac{S}{2}}.
	\end{align}
	Let $\vVec\in\mathbb{K}^N$ be the vector with entries all
	$\eta S^{-1}$ given any $\eta>0$.
	Recall that given $l\in\SetOf{L}$ the vector
	$\ProjToIndex{T_l}{\vVec}$ is the vector with entries
	$\eta S^{-1}$ if $n\in T_l$ and zero else. By
	\eqref{Equation:Theorem:bound_complexity_constant:definition:sets}
	one gets
	\begin{align}
		\label{Equation:Theorem:bound_complexity_constant:v_in_CF}
			\ProjToIndex{T_l}{\vVec}\in C\cap F
			\TextForAll l\in\SetOf{L}
			\textnormal{ and any of the choices of $C,F$ in the assumption.}
	\end{align}
	Further, note that by
	\eqref{Equation:Theorem:bound_complexity_constant:definition:sets}
	for $l\neq l'$ and $q\neq\infty$ one has
	\begin{align}
		\nonumber
			\norm{\ProjToIndex{T_l}{\vVec}-\ProjToIndex{T_{l'}}{\vVec}}_q
		=&
			\left(
				\eta^q S^{-q}\SetSize{T_l}+\eta^q S^{-q}\SetSize{T_{l'}}
				-\eta^q S^{-q}\SetSize{T_l\cap T_{l'}}
			\right)^\frac{1}{q}
		\\\geq&\nonumber
			\left(
				\eta^q S^{1-q}+\eta^q S^{1-q}-\frac{1}{2}\eta^q S^{1-q}
			\right)^\frac{1}{q}
		=
			\left(\frac{3}{2}\right)^\frac{1}{q}\eta S^{\frac{1}{q}-1}.
	\end{align}
	On the other hand by
	\eqref{Equation:Theorem:bound_complexity_constant:definition:sets}
	for $l\neq l'$ and $q=\infty$ one has
	$\norm{\ProjToIndex{T_l}{\vVec}-\ProjToIndex{T_{l'}}{\vVec}}_q=\eta S^{-1}$.
	Thus, for all $q\in\left[1,\infty\right]$
	it holds true that
	\begin{align}
		\label{Equation:Theorem:bound_complexity_constant:v_norms}
			\norm{\ProjToIndex{T_l}{\vVec}}_q
		=
			\eta S^{\frac{1}{q}-1}
		\TextAnd
			\norm{\ProjToIndex{T_l}{\vVec}-\ProjToIndex{T_{l'}}{\vVec}}_q
		\geq
			\left(\frac{3}{2}\right)^\frac{1}{q}\eta S^{\frac{1}{q}-1}
		\TextForAll
			l,l'\in\SetOf{L},l\neq l'.
	\end{align}
	Using \eqref{Equation:Theorem:bound_complexity_constant:v_in_CF}
	and \eqref{Equation:Theorem:bound_complexity_constant:v_norms}
	yields
	\begin{align}
		\nonumber
			\tau\left(\AMat\right)
		=&
			\inf_{\zVec\in F,\xVec\in C:\zVec\neq\xVec}
			\frac{
				\normRHS{\AMat\left(\zVec-\xVec\right)}
			}{
				\norm{\zVec-\xVec}_q
			}
		\leq
			\inf_{l,l'\in\SetOf{L}:l\neq l'}
			\frac{
				\normRHS{
					\AMat\left(
						\ProjToIndex{T_{l'}}{\vVec}
						-\ProjToIndex{T_l}{\vVec}
					\right)
				}
			}{
				\norm{\ProjToIndex{T_{l'}}{\vVec}-\ProjToIndex{T_l}{\vVec}}_q
			}
		\\\leq&\nonumber
			\left(\frac{2}{3}\right)^\frac{1}{q}
			\inf_{l,l'\in\SetOf{L}:l\neq l'}
			\normRHS{
				\AMat\left(
					\eta^{-1} S^{1-\frac{1}{q}}\ProjToIndex{T_{l'}}{\vVec}
					-\eta^{-1} S^{1-\frac{1}{q}}\ProjToIndex{T_l}{\vVec}
				\right)
			}.
	\end{align}
	By setting $
			W
		:=
			\left\{\AMat\eta^{-1}S^{1-\frac{1}{q}}\ProjToIndex{T_l}{\vVec}
				:l\in\SetOf{L}\right\}
	$
	this yields
	\begin{align}
		\label{Equation:Theorem:bound_complexity_constant:tau_distance}
			\tau\left(\AMat\right)
		\leq
			\left(\frac{2}{3}\right)^\frac{1}{q}
			\inf_{\wVec,\wVec'\in W:\wVec\neq\wVec'}
			\normRHS{\wVec-\wVec'}.
	\end{align}
	By \eqref{Equation:Theorem:bound_complexity_constant:v_in_CF},
	$0\in F$ and the definition of $\kappa\left(\AMat\right)$
	one gets
	\begin{align}
		\nonumber
			\normRHS{\AMat\eta^{-1}S^{1-\frac{1}{q}}\ProjToIndex{T_l}{\vVec}}
		=
			\eta^{-1}S^{1-\frac{1}{q}}
			\normRHS{\AMat\left(\ProjToIndex{T_l}{\vVec}-0\right)}
		\leq
			\eta^{-1}S^{1-\frac{1}{q}}
			\kappa\left(\AMat\right)
			\norm{\ProjToIndex{T_l}{\vVec}-0}_q.
	\end{align}
	Using
	\eqref{Equation:Theorem:bound_complexity_constant:v_norms} 
	on this	yields
	\begin{align}
		\nonumber
			\normRHS{\AMat\eta^{-1}S^{1-\frac{1}{q}}\ProjToIndex{T_l}{\vVec}}
		\leq
			\kappa\left(\AMat\right)
		\leq
			1.
	\end{align}
	Hence, $W\subset \mathbb{B}_{\normRHS{\cdot}}^M$. Applying
	\thref{Lemma:distance_in_unit_ball} to
	\eqref{Equation:Theorem:bound_complexity_constant:tau_distance}
	yields
	\begin{align}
		\nonumber
			\tau\left(\AMat\right)
		\leq
			2\left(\frac{2}{3}\right)^\frac{1}{q}
			\left(\Exp{\frac{1}{c'M}\Ln{L}}-1\right)^{-1}.
	\end{align}
	Plugging
	\eqref{Equation:Theorem:bound_complexity_constant:definition:sets}
	into this gives
	\begin{align}
		\nonumber
			\tau\left(\AMat\right)
		\leq
			2\left(\frac{2}{3}\right)^\frac{1}{q}
			\left(\Exp{\frac{1}{c'M}\Ln{
				\left(\frac{N}{4S}\right)^{\frac{S}{2}}
			}}-1\right)^{-1}
		=
			2\left(\frac{2}{3}\right)^\frac{1}{q}
			\left(\Exp{\frac{S}{cM}\Ln{\frac{N}{4S}}}-1\right)^{-1}.
	\end{align}
	Since this holds for all $\AMat$ with $\kappa\left(\AMat\right)\leq 1$,
	the claim follows.
\end{proof}
\begin{Remark}\label{Remark:compelexity_constant:scaling}
	In the cases of \thref{Theorem:bound_complexity_constant}
	the robustness constant can only be constant over dimensions if
	$M\in\mathcal{O}\left(S\Ln{\frac{N}{S}}\right)$. Otherwise
	the robustness constant changes exponentially in the variable
	$\frac{S}{cM}\Ln{\frac{N}{4S}}$.
	In the case $M\in\mathcal{O}\left(S\right)$ the robustness constant
	can be finite but grows with growing $N$.
\end{Remark}
In many compressed sensing works there is an emphasis on stability
and robustness, see \cite[Chapter~4]{IntroductionCS}.
Stability is closely related to a quantity called best $S$-term approximation error, see \cite[Definition~2.2]{IntroductionCS} for a precise definition.
While robustness means bounding the error by
$\normRHS{\yVec-\AMat\xVec}$, stability means bounding the error by the best $S$-term approximation error \cite[Theorem~4.22]{IntroductionCS}.
In this work the robustness was investigated and several results have been shown.
At least in the case of $\mathbb{K}=\mathbb{R}$, $C=\Sigma_S^N\cap\mathbb{R}^N_+$, $F=\mathbb{R}^N_+$
and $\normLHS{\cdot}=\norm{\cdot}_1$ the best $S$-term approximation error
of $\xVec$ is given by the quantity $\inf_{\zVec\in C}\normLHS{\zVec-\xVec}$
from \thref{Theorem:robustness=>stability}. In this case stability can
be considered as a consequence of robustness.
\par
Stability and the best $S$-term approximation are also closely related
to $\ell_p$-instance optimality, see \cite[Definition~11.1]{IntroductionCS}. In particular, in
\cite[Theorem~11.6]{IntroductionCS} it was proven that constant
stability can only hold if $M\geq CS\Ln{\ExpE\frac{N}{S}}$.
To be precise, \cite[Proof~of~Theorem~11.6]{IntroductionCS} reveals that
any constant $C$ of a stable recovery guarantee must obey
$C\geq \frac{1}{2}\Exp{\frac{1}{4}\frac{S}{M}\Ln{\ExpE\frac{N}{S}}-2}-\frac{3}{2}$.
A similar scaling also holds for \cite[Theorem~11.7]{IntroductionCS}.
In this work a similar statement about the robustness constant
was proven with
\thref{Remark:compelexity_constant:scaling} and
\thref{Theorem:bound_complexity_constant}.
In view of
\thref{Theorem:robustness=>stability} the question arises
whether the stability gets worse because the robustness gets worse
or vice versa or each happens independently in the case
$M\notin \mathcal{O}\left(S\Ln{\frac{N}{S}}\right)$.
	\section*{Acknowledgments}
The authors would like to thank Professor Philippe Ciblat for translating 
\cite{complex_rule_of_signs_french}
and Fabian Jaensch for the discussions about the complex null space property case of
\thref{Proposition:cone_closed}.
	\bibliography{Postamble/Bibliography}
	\bibliographystyle{IEEEtran}
\end{document}